\documentclass[onecolumn, a4size, 11pt]{IEEEtran}
\usepackage{amsmath}
\usepackage{amssymb}
\usepackage{amsfonts}
\usepackage{graphicx}
\usepackage{epsfig}
\usepackage{subfigure}
\usepackage{psfrag}
\renewcommand{\baselinestretch}{1.4}

\title{Exploiting Multi-Antennas for Opportunistic Spectrum Sharing in Cognitive Radio Networks
\footnote{Manuscript received on May 18, 2007; revised on August 29,
2007. This paper was presented in part at IEEE International
Symposium on Personal, Indoor and Mobile Radio Communications
(PIMRC), Athens, Greece, September 2007.} \footnote {The authors are
with the Institute for Infocomm Research, 21 Heng Mui Keng Terrace,
119613 Singapore (e-mails:\{rzhang, ycliang\}@i2r.a-star.edu.sg).}}

\author{Rui Zhang, {\it Member, IEEE,} Ying-Chang Liang, {\it Senior Member,
IEEE}}

\begin{document}
\maketitle \maketitle \thispagestyle{empty}

\begin{abstract}
In cognitive radio (CR) networks, there are scenarios where the
secondary (lower priority) users intend to communicate with each
other by opportunistically utilizing the transmit spectrum
originally allocated to the existing primary (higher priority)
users. For such a scenario, a secondary user usually has to trade
off between two conflicting goals at the same time: one is to
maximize its own transmit throughput; and the other is to minimize
the amount of interference it produces at each primary receiver.

In this paper, we study this fundamental tradeoff from an
information-theoretic perspective by characterizing the secondary
user's channel capacity under both its own transmit-power constraint
as well as a set of interference-power constraints each imposed at
one of the primary receivers. In particular, this paper exploits
multi-antennas at the secondary transmitter to effectively balance
between spatial multiplexing for the secondary transmission and
interference avoidance at the primary receivers. Convex optimization
techniques are used to design algorithms for the optimal secondary
transmit spatial spectrum that achieves the capacity of the
secondary transmission. Suboptimal solutions for ease of
implementation are also presented and their performances are
compared with the optimal solution. Furthermore, algorithms
developed for the single-channel transmission are also extended to
the case of multi-channel transmission whereby the secondary user is
able to achieve opportunistic spectrum sharing via transmit
adaptations not only in space, but in time and frequency domains as
well. Simulation results show that even under stringent
interference-power constraints, substantial capacity gains are
achievable for the secondary transmission by employing
multi-antennas at the secondary transmitter. This is true even when
the number of primary receivers exceeds that of secondary transmit
antennas in a CR network, where an interesting ``interference
diversity'' effect can be exploited.
\end{abstract}

\vspace{0.5in}

\newpage

\setlength{\baselineskip}{1.3\baselineskip}
\newtheorem{claim}{Claim}
\newtheorem{guess}{Conjecture}
\newtheorem{definition}{Definition}
\newtheorem{fact}{Fact}
\newtheorem{assumption}{Assumption}
\newtheorem{theorem}{Theorem}
\newtheorem{lemma}{Lemma}
\newtheorem{ctheorem}{Corrected Theorem}
\newtheorem{corollary}{Corollary}
\newtheorem{proposition}{Proposition}
\newtheorem{example}{Example}
\newtheorem{remark}{Remark}[section]
\newtheorem{problem}{\underline{Problem}}
\newtheorem{algorithm}{\underline{Algorithm}}
\newcommand{\mv}[1]{\mbox{\boldmath{$ #1 $}}}

\section{Introduction}
Fixed spectrum allocation is the major spectrum allocation
methodology for traditional wireless communication services. In
particular, in order to avoid interference, different wireless
services are allocated with different licensed bands. With the
popularity of various wireless technologies, fixed spectrum
allocation strategy has resulted in scarcity in radio spectrum, due
to the fact that most of the available spectrum has been allocated.
Because of the spectrum scarcity, one immediate consequence is that,
while there are a lot of research activities dealing with technical
issues related to the fourth-generation (4G) cellular systems, it is
still unclear which frequency band is available for such systems. On
the other hand, recent measurements by FCC and others have shown
that more than 70\% of the allocated spectrum in United States is
not utilized \cite{FCC03}. Furthermore, the spectrum utilization
varies in space, time and frequency. This motivates the invention of
cognitive radio (CR) network \cite{Mitola99}-\cite{Haykin05}, which
supports \emph{opportunistic spectrum sharing} by allowing the
secondary (lower priority) users to share the radio spectrum
originally allocated to the primary (higher priority) users. By
doing so, the utilization efficiency of the radio spectrum can be
significantly enhanced.

In CR networks, the primary users have a higher priority than the
secondary users in utilizing spectrum resource, therefore, one
fundamental challenge by introducing opportunistic spectrum sharing
is to ensure the quality-of-service (QoS) of the primary users while
maximizing the achievable throughput of the secondary users.
Recently, there has been a great deal of research related to this
interesting problem. When the primary users are legacy systems that
do not actively participate in transmit power control, the QoS of
the primary users is maintained by introducing interference-power
constraint measured at the primary receiver. That is to say, the
interference power received at the primary receiver should be less
than a threshold. Under this setup, a seminal work in
\cite{Gastpar07} studied the channel capacity of a single secondary
transmission under a set of receiver-side power constraints instead
of the conventional transmitter-side power constraints. In
\cite{Ghasemi07}, the ergodic capacity of the secondary transmission
link in a fading environment is studied under instantaneous or
average interference-power constraint at a single primary receiver.
When multiple secondary users share a single-frequency channel
advertised by the primary user, the authors in
\cite{Huang06}-\cite{Xing07} utilize the game theory to maximize the
sum of the utility functions of the secondary users under the
interference-power constraint at some measured point, while in
\cite{Zheng05}-\cite{Wang05} transmit resource allocation for the
secondary users is studied by applying the graph-theoretic models.
Recent information-theoretic studies on CR networks can also be
found in, e.g., \cite{Syed07}-\cite{Syed07a}.

Most prior research on radio resource allocation for CR networks
focuses on time and/or frequency domains, while assuming single
antenna employed at both primary and secondary transceivers.
Wireless transmissions via multiple transmit antennas and multiple
receive antennas, or the so-called multiple-input multiple-output
(MIMO) transmissions, have received considerable attention during
the past decade. Multi-antennas can be utilized to achieve many
desirable functions for wireless transmissions, such as folded
capacity increase without bandwidth expansion (e.g.,
\cite{Telatar99}-\cite{Raleigh98}), dramatic enhancement of
transmission reliability via space-time coding (e.g.,
\cite{Tarokh98}-\cite{Tarokh99}), and effective co-channel
interference suppression for multiuser transmissions (e.g.,
\cite{Liu98}). However, the role of multi-antennas in a CR network
is yet completely understood. Generally speaking, multi-antennas can
be used to allocate transmit dimensions in space and hence provide
the secondary transmitter in a CR network more degrees of freedom in
space in addition to time and frequency so as to balance between
maximizing its own transmit rate and minimizing the interference
powers at the primary receivers. This motivates the research of this
paper to be done, with an aim to address fundamentally the MIMO
channel capacity of a secondary user under optimum spectrum sharing
in a CR network.

The main contributions of this paper are summarized as follows:
\begin{itemize}
\item This paper formulates the design of capacity-achieving transmit spatial spectrum
for a single secondary link in a CR network under both its own
transmit-power constraint and a set of interference-power
constraints at the primary receivers as a sequence of {\it convex}
optimization problems. Thanks to convexity of these formulated
problems, efficient numerical algorithms are proposed to obtain the
optimal secondary transmit spatial spectrum for any arbitrary number
of secondary transmit and receive antennas, as well as primary
receivers each having single or multiple antennas.

\item In the case where the secondary user's channel is
multiple-input single-output (MISO), i.e., there is only a single
antenna at the secondary receiver, this paper proves that {\it
beamforming} is the optimal strategy for the secondary transmission.
For the special case where there is only one single-antenna primary
receiver, we are able to derive the closed-form solution for the
optimal beamforming vector at the secondary transmitter. For the
more general MIMO case where multiple antennas are equipped at both
the secondary transmitter and receiver, this paper presents two
suboptimal algorithms to tradeoff between {\it spatial multiplexing}
for the secondary transmission and interference avoidance at the
primary receivers. One algorithm is based on the singular-value
decomposition (SVD) of the secondary user's MIMO channel directly
and is thus referred to as the {\it Direct-Channel SVD} (D-SVD); and
the other is based on the SVD of the secondary MIMO channel after
the projection into the null space of the channel from the secondary
transmitter to the primary receivers (thereby removing completely
the interference at all primary receivers) and is thus referred to
as the {\it Projected-Channel SVD} (P-SVD).

\item In the case of multiple
primary receivers with single or multiple antennas, a {\it hybrid}
D-SVD/P-SVD algorithm is proposed to remove partially, as opposed to
completely in the case of P-SVD, the interferences at the primary
receivers. An interesting and novel {\it interference diversity}
effect is discovered and exploited.

\item At last, this paper extends the developed algorithms for the single-channel transmission to
the case of multi-channel transmission whereby the secondary
transmitter is able to adapt transmit resources like transmit
spatial spectrum, power and rate in both space and time/freqeuncy
for opportunistic spectrum sharing. This paper shows that the
multi-channel resource allocation problem can be efficiently solved
using the {\it Lagrange dual-decomposition method} that decomposes
the original problem into a set of smaller-size subproblems each for
one of the sub-channels.
\end{itemize}

In \cite{Gastpar07}, Gastpar also considered the MIMO channel
capacity for the spectrum sharing scenario. Since only the
interference-power constraint is considered in \cite{Gastpar07}, it
requires that there exist multiple primary receivers/antennas such
that the channel matrix from the secondary transmitter to the
primary receivers/antennas is invertible, and by doing so, an
upper-bound of the secondary user's MIMO channel capacity is derived
by transforming the set of receiver-side power constraints into a
transmitter-side power constraint. In contrast, in this paper we
consider the MIMO channel capacity for the secondary user under both
the interference-power constraints at the primary receivers as well
as an explicit transmit-power constraint for the secondary user.
With addition of this explicit transmitter-side power constraint, we
are able to quantify the exact channel capacity for the secondary
user for any arbitrary number of primary receivers/antennas.

This paper is organized as follows. Section \ref{sec:system model}
provides the system model of a CR network under opportunistic
spectrum sharing, and presents the general problem formulation for
transmit optimization of the secondary user. Section \ref{sec:single
primary receiver} presents the solution for the simplest case where
there is only one single-antenna primary receiver, and both primary
and secondary users share the same single-channel for transmission.
Section \ref{sec:multiple primary receivers} and Section
\ref{sec:multi-channel} extend the solutions to incorporate multiple
primary receivers/antennas and the multi-channel transmission,
respectively. Section \ref{sec:simulation results} provides the
simulation results. Finally, Section \ref{sec:conclusion} concludes
the paper.

The following notations are used in this paper. $|\mv{S}|$ denotes
the determinant, and $\mathtt{Tr}(\mv{S})$ the trace of a square
matrix $\mv{S}$, and $\mv{S}\succeq 0$ means that $\mv{S}$ is a
positive semi-definite matrix. For any general matrix $\mv{M}$,
$\mv{M}^{\dag}$ denotes its conjugate transpose, and
$\mathtt{Rank}(\mv{M})$ denotes its rank. $\mv{I}$ denotes the
identity matrix. $\mathbb{E}[\cdot]$ denotes statistical
expectation. $\mathbb{C}^{x \times y}$ denotes the space of $x\times
y$ matrices with complex entries. The distribution of a
circularly-symmetric-complex-Gaussian (CSCG) vector with the mean
vector $\mv{x}$ and the covariance matrix $\mv{\Sigma}$ is denoted
by $\mathcal{CN}(\mv{x},\mv{\Sigma})$, and $\sim$ means
``distributed as''. The quantity $\min(x,y)$ and $\max(x,y)$ denote,
respectively, the minimum and the maximum between two real numbers,
$x$ and $y$, and $(x)^+\triangleq\max(x,0)$. The quantity
$||\mv{x}||$ denotes the Euclidean norm of a vector (also for a
scalar) $\mv{x}$, i.e.,
$||\mv{x}||^2=\mathtt{Tr}(\mv{x}\mv{x}^{\dag})$. $\mathtt{Re}(x)$
and $\mathtt{Im}(x)$ denote the real and imaginary part of a complex
number $x$, respectively.

\section{System Model and Problem Formulation} \label{sec:system model}

This paper considers a CR network with $K$ primary receivers and a
single pair of secondary transmitter and receiver as shown in Fig.
\ref{fig:CR system}. It is assumed that all the primary users and
the secondary user share the same bandwidth for transmission. This
paper considers the scenario where multiple antennas are equipped at
the secondary transmitter and possibly at the secondary receiver and
each of the primary receivers. It is assumed that the MIMO/MISO
channels from the secondary transmitter to the secondary and primary
receivers are perfectly known at the secondary transmitter. Under
such assumptions, the secondary transmitter is able to adapt its
transmit resources such as transmit rate, power, and spatial
spectrum based upon the channel knowledge so as to optimally balance
between maximizing its own transmit throughput and avoiding
interferences at the primary receivers. In practice, the channel
from the secondary transmitter to the primary receiver can be
obtained at the secondary transmitter by, e.g., periodically sensing
the transmitted signal from the primary receiver provided that
time-division-duplex (TDD) is employed by the primary transmission.
In a fading environment, there are cases where it is difficult for
the secondary transmitter to perfectly estimate instantaneous
channels. In such cases, the results obtained in this paper provide
capacity upper-bounds for the secondary transmission in a CR
network.

Consider first a single-channel transmission (e.g., narrow-band
transmission with deterministic channels) for both primary and
secondary users. The extension to the multi-channel transmission
(e.g., multi-tone transmission in frequency or consecutive
block-fading channels in time) is considered later in Section
\ref{sec:multi-channel} of this paper. In the single-channel case,
the secondary transmission can be represented by
\begin{equation}\label{eq:channel model}
\mv{y}(n) = \mv{H}\mv{x}(n) + \mv{z}(n).
\end{equation}
In the above, $\mv{H}\in\mathbb{C}^{M_{r,s} \times M_{t,s}}$ denotes
the secondary user's channel (assumed to be full rank) where
$M_{r,s}$ and $M_{t,s}$ are the number of antennas at the secondary
receiver and transmitter, respectively; $\mv{y}(n)$ and $\mv{x}(n)$
are the received and transmitted signal vector, respectively, and
$n$ is the symbol index; $\mv{z}(n)$ is the additive noise vector at
the secondary receiver, and it is assumed that
$\mv{z}(n)\sim\mathcal{CN}(0,\mv{I})$.\footnote{The noise at the
secondary receiver may also contain the interference from the
primary transmitters (not shown in Fig. \ref{fig:CR system}) and is
thus non-white in general. However, by applying a noise-whitening
filter at the secondary receiver and incorporating this filter
matrix into the channel matrix $\mv{H}$, the equivalent noise at the
secondary receiver can be assumed to be approximately white
Gaussian.} Let the transmit covariance matrix (or spatial spectrum)
of the secondary user be denoted by $\mv{S}$,
$\mv{S}=\mathbb{E}[\mv{x}(n)\mv{x}^{\dag}(n)]$, where the
expectation is taken over the code-book. Since this paper
characterizes the information-theoretic limit of the secondary
transmission, it is assumed that the ideal Gaussian code-book with
infinitely large number of codeword symbols is used, i.e.,
$\mv{x}(n)\sim\mathcal{CN}(0,\mv{S}), n=1,2,\ldots,\infty$. The
transmit covariance matrix $\mv{S}$ can be further represented by
its eigenvalue decomposition expressed as
\begin{equation}\label{eq:SVD}
\mv{S}=\mv{V}\mv{\Sigma}\mv{V}^{\dag},
\end{equation}
where $\mv{V}\in\mathbb{C}^{M_{t,s}\times d}$,
$\mv{V}^{\dag}\mv{V}=\mv{I}$, contains the eigenvectors of $\mv{S}$,
and is also termed in practice as the {\it precoding matrix} because
each column of $\mv{V}$ is the precoding vector for one transmitted
data stream; $d$, $d\leq M_{t,s}$, is usually referred to as the
degree of {\it spatial multiplexing} because it measures the number
of transmit dimensions (or equivalently, the number of data streams)
in the spatial domain; $\mv{\Sigma}$ is a $d \times d$ diagonal
matrix, and its diagonal elements, denoted by
$\sigma_1,\sigma_2,\ldots,\sigma_{d}$, are the positive eigenvalues
of $\mv{S}$, and also represent the assigned transmit powers for
their corresponding data streams. Notice that
$\mathtt{Rank}(\mv{S})=d$. If $d=1$, the corresponding transmission
strategy is usually termed as {\it beamforming} while in the case of
$d>1$, it is termed as {\it spatial multiplexing}. The transmit
power $P$ for each block is limited by the secondary user's own
transmit power constraint denoted by $P_t$, i.e., it holds that
$P=\mathtt{Tr}(\mv{S})=\sum_{i=1}^{d}\sigma_i\leq P_t$.

Assuming that there are $K$ primary receivers in the CR network,
each equipped with $M_k$ receive antennas, $k=1,\ldots,K$. For each
primary receiver, there may be a total interference-power constraint
over all receive antennas or a set of interference-power constraints
applied to each individual receive antenna. The former case can be
expressed as
\begin{equation}\label{eq:total power ct}
\sum_{j=1}^{M_k}\mv{g}_{k,j}\mv{S}\mv{g}_{k,j}^{\dag}\leq \Gamma_k,
\ \ k=1,\ldots,K,
\end{equation}
where $\mv{g}_{k,j}\in\mathbb{C}^{1\times M_{t,s}}$ represents the
channel from the secondary transmitter to the $j$-th receive antenna
of the $k$-th primary receiver, and $\Gamma_k$ is the total
interference-power constraint over all receive antennas for the
$k$-th primary receiver. Let $\mv{G}_k\in\mathbb{C}^{M_k\times
M_{t,s}}$ (assumed to be full-rank) be the equivalent channel matrix
from the secondary transmitter to the $k$-th primary receiver
obtained by stacking all $\mv{g}_{k,j}$, $j=1,\ldots,M_k$, into
$\mv{G}_k$. Using $\mv{G}_k$'s, (\ref{eq:total power ct}) can be
thus rewritten as
\begin{equation}\label{eq:total power ct 2}
\mathtt{Tr}\left(\mv{G}_k\mv{S}\mv{G}_k^{\dag}\right) \leq \Gamma_k,
\ \ k=1,\ldots,K.
\end{equation}
The latter case can be represented as
\begin{equation}\label{eq:indiv power ct}
\mv{g}_{k,j}\mv{S}\mv{g}_{k,j}^{\dag}\leq \gamma_k, \ \
j=1,\ldots,M_k, k=1,\ldots,K,
\end{equation}
where $\gamma_k$ is the interference-power constraint applied to
each antenna of the $k$-th primary receiver, and is assumed to be
identical for all of its receive antennas. Notice that if
$\Gamma_k=M_k\gamma_k$, the per-antenna power constraint in
(\ref{eq:indiv power ct}) is more stringent than the total-power
constraint in (\ref{eq:total power ct}). On the other hand,
(\ref{eq:indiv power ct}) can be considered as a special case of
(\ref{eq:total power ct}) because if each receive antenna is treated
as an independent primary receiver, (\ref{eq:indiv power ct}) is
equivalent to (\ref{eq:total power ct}) if a total number of
$M_{r,p}=\sum_{k=1}^KM_k$ single-antenna primary receivers are
assumed. Therefore, for the rest of this paper, the total
interference-power constraint in (\ref{eq:total power ct}) or
(\ref{eq:total power ct 2}) is considered as the generalized
interference-power constraint at each primary receiver.

The use of total interference-power constraint at each primary
receiver can be further justified by the following theorem:
\begin{theorem}\label{theorem:capacity loss}
The capacity loss of the $k$-th primary transmission,
$k=1,2,\ldots,K$, due to the interference from the secondary
transmitter with transmit covariance matrix $\mv{S}$ that satisfies
the set of total interference-power constraints in (\ref{eq:total
power ct 2}) is upper-bounded by
$\min(M_k,N_k)\log_2\left(1+\frac{\Gamma_k}{\phi_k}\right)$
bits/complex dimension, where $N_k$ and $M_k$ are the number of
transmit and receiver antennas for the $k$-th primary user,
respectively, and $\phi_k$ is the additive white Gaussian noise
power at its receiver.
\end{theorem}
\begin{proof}
Please refer to Appendix \ref{appendix:capacity loss}.
\end{proof}
From Theorem \ref{theorem:capacity loss}, it follows that by
imposing the total interference-power constraint, it is ensured that
the capacity loss of each primary user due to the secondary
transmission is well regulated. Notice that the obtained upper-bound
for the capacity loss is not a function of the primary/secondary
user's channels, or their transmit covariance matrices. By choosing
$\Gamma_k$ to be sufficiently small compared with the noise power at
each primary receiver, e.g., $\Gamma_k\ll\phi_k$, the capacity loss
due to the secondary transmission can be made arbitrarily small.

At last, we formulate the main problem to be addressed in this
paper. We are interested in the design of the spatial spectrum
$\mv{S}$ at the secondary transmitter so as to maximize its transmit
rate under both its own transmit-power constraint and a set of total
interference-power constraints at $K$ primary receivers.
Accordingly, the optimal $\mv{S}$ can be obtained by solving the
following problem ({\bf P1}):
\begin{eqnarray}
\mathtt{Maximize}&&
\log_2\left|\mv{I}+\mv{H}\mv{S}\mv{H}^{\dag}\right|
\label{eq:capacity P1}
\\ \mathtt {Subject \ to} && \mathtt{Tr}(\mv{S})\leq P_t, \label{eq:Tx power P1}\\
&& \mathtt{Tr}\left(\mv{G}_k\mv{S}\mv{G}_k^{\dag}\right)\leq
\Gamma_k, \ k=1,\ldots,K \label{eq:Ic power P1} \\ && \mv{S}\succeq
0. \label{eq:S P1}
\end{eqnarray}
The above problem maximizes the secondary user's channel capacity
(in bits/complex dimension) obtained by computing the mutual
information \cite{Coverbook} between the channel input and output in
(\ref{eq:channel model}), assuming that the secondary user's channel
is also known perfectly at the secondary receiver. The constraints
in (\ref{eq:Tx power P1}) and (\ref{eq:Ic power P1}) correspond to
the secondary transmitter power constraint and the
interference-power constraints at the primary receivers,
respectively. (\ref{eq:S P1}) is due to the fact that the spectrum
matrix $\mv{S}$ must be positive semi-definite. The problem at hand
is a convex optimization problem \cite{Boydbook} because its
objective function is a concave function of $\mv{S}$ and all of its
constraints specify a convex set of $\mv{S}$. Therefore, for any
arbitrary number of secondary transmit and receive antennas, and
primary receivers each having single or multiple antennas, this
problem can be efficiently solved by using standard convex
optimization techniques, e.g., the interior-point method
\cite{Boydbook}, for which the details are omitted here for brevity.
In the sequel, we will investigate further into this problem so as
to provide more insightful solutions for it that may not be
obtainable from solely a numerical optimization perspective. We will
first consider the simplest case where there is only one
single-antenna primary receiver at present in Section
\ref{sec:single primary receiver}, and then consider the general
case of multiple primary receivers/antennas in Section
\ref{sec:multiple primary receivers}.

\section{One Single-Antenna Primary Receiver}\label{sec:single primary receiver}

The scenario where there is only one single-antenna primary receiver
is considered in this section. Consequently, the channel from the
secondary transmitter to the primary receiver is MISO, and can be
represented by a vector $\mv{g}\in\mathbb{C}^{1 \times M_{t,s}}$.
Let $\gamma$ denote the maximum interference power tolerable at the
primary receiver. Problem {\bf P1} can be then simplified as ({\bf
P2})
\begin{eqnarray}
\mathtt{Maximize}&&
\log_2\left|\mv{I}+\mv{H}\mv{S}\mv{H}^{\dag}\right|
\label{eq:capacity}
\\ \mathtt {Subject \ to} && \mathtt{Tr}(\mv{S})\leq P_t, \label{eq:Tx power}\\  && \mv{g}\mv{S}\mv{g}^{\dag}\leq
\gamma, \label{eq:Ic power} \\ && \mv{S}\succeq 0. \label{eq:S}
\end{eqnarray}

\subsection{MISO Secondary User's Channel}
Consider first the case where there is only a single antenna at the
secondary receiver, i.e., $M_{r,s}=1$ and the secondary user's
channel is also MISO. Hence, $\mv{H}$ is in fact a vector and for
convenience is denoted by $\mv{H}\equiv\mv{h}$,
$\mv{h}\in\mathbb{C}^{1 \times M_{t,s}}$. In this case, we are able
to derive the closed-form solution for the optimal $\mv{S}$. First,
the following two lemmas are needed:
\begin{lemma}\label{lemma:MISO 1}
In the case of MISO secondary user's channel, the optimal $\mv{S}$
for Problem {\bf P2} satisfies $\mathtt{Rank}(\mv{S})=1$.
\end{lemma}
\begin{proof}
Please refer to Appendix \ref{appendix:proof of lemma 1}.
\end{proof}

Lemma \ref{lemma:MISO 1} indicates that in the MISO case, {\it
beamforming} is indeed optimal for the secondary transmitter.
Therefore, $\mv{S}$ can be written in the form of
$\mv{S}=\mv{v}\mv{v}^{\dag}, \mv{v}\in\mathbb{C}^{M_{t,s}\times 1}$.
Problem {\bf P2} can be then simplified as
\begin{eqnarray}
\mathtt{Maximize}&& \log_2\left(1+||\mv{h}\mv{v}||^2\right)
\label{eq:capacity MISO}
\\ \mathtt {Subject \ to} && ||\mv{v}||^2\leq P_t, \label{eq:Tx power MISO}\\  && ||\mv{g}\mv{v}||^2\leq
\gamma. \label{eq:Ic power MISO}
\end{eqnarray}
The second lemma then provides the optimal structure of $\mv{v}$:
\begin{lemma}\label{lemma:MISO 2}
In the case of MISO secondary user's channel, the optimal
beamforming vector $\mv{v}$ is in the form of $\alpha_v
\hat{\mv{g}}+ \beta_v \hat{\mv{h}}_{\bot} $, where
$\hat{\mv{g}}=\frac{\mv{g}^{\dag}}{||\mv{g}||}$ and
$\hat{\mv{h}}_{\bot}=\frac{\mv{h}_{\bot}}{||\mv{h}_{\bot}||}$,
$\mv{h}_{\bot}=\mv{h}^{\dag}-(\hat{\mv{g}}^{\dag}\mv{h}^{\dag})\hat{\mv{g}}$,
$\alpha_v$ and $\beta_v$ are complex weights.
\end{lemma}
\begin{proof}
Please refer to Appendix \ref{appendix:proof of lemma 2}.
\end{proof}

Lemma \ref{lemma:MISO 2} states that the optimal beamforming vector
should lie in the space spanned jointly by $\mv{g}^{\dag}$ and the
projection of $\mv{h}^{\dag}$ into the null space of
$\mv{g}^{\dag}$. Using Lemma \ref{lemma:MISO 2} and let
$\mv{h}^{\dag}=\alpha_h \hat{\mv{g}} + \beta_h \hat{\mv{h}}_{\bot}$,
the optimal weights $\alpha_v$ and $\beta_v$ can be then obtained by
considering the following equivalent problem of that in
(\ref{eq:capacity MISO})-(\ref{eq:Ic power MISO}):
\begin{eqnarray}
\mathtt{Maximize}&&
||\alpha_h^{\dag}\alpha_v+\beta_h^{\dag}\beta_v||^2
\\ \mathtt {Subject \ to} && ||\alpha_v||^2+ ||\beta_v||^2\leq P_t, \\  && ||\mv{g}||^2||\alpha_v||^2\leq
\gamma.
\end{eqnarray}
The above problem can be solved by using standard geometry. Notice
that the solution for the above problem is trivial in the case of
$\alpha_h=0$ or $\beta_h=0$. Hence, without loss of generality, it
is assumed that $\alpha_h\neq0$ and $\beta_h\neq0$. To summarize,
the following theorem is established:
\begin{theorem}\label{theorem:MISO}
In the case of MISO secondary user's channel, the optimal transmit
covariance matrix $\mv{S}$ for Problem {\bf P2} can be written in
the form of $\mv{S}=\mv{v}\mv{v}^{\dag}$, where $\mv{v}=\alpha_v
\hat{\mv{g}}+ \beta_v \hat{\mv{h}}_{\bot} $, and $\alpha_v$ and
$\beta_v$ are given by
\begin{itemize}
\item Case I: If $\gamma \geq
\frac{||\mv{g}||^2||\alpha_h||^2}{||\alpha_h||^2+||\beta_h||^2}P_t$,
\begin{equation}
\alpha_v=\sqrt{\frac{P_t}{||\alpha_h||^2+||\beta_h||^2}}\alpha_h, \
\ \beta_v=\sqrt{\frac{P_t}{||\alpha_h||^2+||\beta_h||^2}}\beta_h.
\nonumber
\end{equation}
\item Case II: If $\gamma <
\frac{||\mv{g}||^2||\alpha_h||^2}{||\alpha_h||^2+||\beta_h||^2}P_t$,
\begin{eqnarray*}
\alpha_v=\frac{\sqrt{\gamma}}{||\mv{g}||}\frac{\alpha_h}{||\alpha_h||},
\ \beta_v =
  \sqrt{P_t-\frac{\gamma}{||\mv{g}||^2}}\frac{\beta_h}{||\beta_h||}.
\end{eqnarray*}
\end{itemize}
\end{theorem}

In Theorem \ref{theorem:MISO}, Case I corresponds to the
interference-power constraint in (\ref{eq:Ic power MISO}) being
inactive. In this case, it can be verified that
$\mv{v}=\frac{\sqrt{P_t}}{||\mv{h}||}\mv{h}^{\dag}$. Therefore, the
optimal beamforming vector is identical to that obtained by the
pre-maximal-ratio-combining (MRC) principle for the conventional
MISO point-to-point transmission. On the other hand, Case II
corresponds to the interference-power constraint being active and,
hence, the transmit power $||\alpha_v||^2$ allocated in the
direction of $\hat{\mv{g}}$ needs to be regulated by $\gamma$.

\subsection{MIMO Secondary User's Channel} \label{subsec:MIMO}
In the case that the secondary user's channel is MIMO, i.e.,
$M_{t,s}>1$ and $M_{r,s}>1$,  there is in general no closed-form
solution for the optimal $\mv{S}$, and Problem {\bf P2} needs to be
solved numerically. Unlike the MISO case, it is possible that
$\mathtt{Rank}(\mv{S})>1$ in the MIMO case, which implies that {\it
spatial multiplexing} is optimal instead of beamforming. Without the
interference-power constraint in (\ref{eq:Ic power}), the optimal
$\mv{S}$ for Problem {\bf P2} can be obtained from the SVD of
$\mv{H}$, along with the water-filling (WF) -based power allocation
(e.g., \cite{Telatar99}, \cite{Coverbook}). The SVD-based
transmission is not only capacity-achieving, but of more practical
significance, it diagonalizes the secondary user's MIMO channel
matrix and decomposes it into parallel AWGN sub-channels, and
thereby, reduces substantially the overall encoding/decoding
complexity since independent encoding/decoding can be applied over
these sub-channels. Unfortunately, when the interference-power
constraint in (\ref{eq:Ic power}) is applied, the optimal $\mv{S}$
is in general different from that obtained from the channel SVD, and
hence does not diagonalize the secondary user's MIMO channel. As a
result, sophisticated encoding and decoding methods need to be used
in order to achieve the secondary user's channel capacity. One
capacity-achieving method is to use a {\it single constant-rate}
Gaussian code-book \cite{Caire99} to encode the information bits and
then spread the coded symbols to all transmit antennas. At the
receiver, the maximum-likelihood (ML) -based detection and iterative
decoding (e.g., \cite{Brink03}) are applied to decode the whole
codeword. Alternatively, the capacity can also be achieved by using
{\it multiple variable-rate} Gaussian code-books each for one of $d$
data streams. At the receiver, these data streams are decoded by
optimum decision-feedback-based successive decoding (e.g.,
\cite{Cioffi97}-\cite{Foschini96}).

This paper presents two suboptimal algorithms that obtain
closed-form solutions for $\mv{S}$ in the case of MIMO secondary
user's channel. The computational complexity for both algorithms is
much lower than that of the interior-point method for Problem {\bf
P2}. Furthermore, both algorithms are based on the SVD of the
secondary user's channel matrix and decompose the secondary MIMO
channel into parallel sub-channels. The main difference between
these two algorithms lies in that, in the first one, the channel
decomposition is based on the SVD of $\mv{H}$ directly and is thus
referred to as the {\it Direct-Channel SVD} (D-SVD), while in the
second one, the channel decomposition is based on the SVD of the
projection of $\mv{H}$ into the null space of $\mv{g}^{\dag}$, and
is thus referred to as the {\it Projected-Channel SVD} (P-SVD).
Notice that these two algorithms can also be applied in the previous
case of MISO secondary user's channel to obtain closed-form (but in
general suboptimal) solutions for transmit beamforming vector
$\mv{v}$.

\subsubsection{Direct-Channel SVD (D-SVD)}

The eigenvalue decomposition of $\mv{S}$ is represented by the
precoding matrix $\mv{V}$ and the power allocation $\mv{\Sigma}$ in
(\ref{eq:SVD}). In the D-SVD, the precoding matrix $\mv{V}$ is
obtained from the SVD of $\mv{H}$, which can be expressed as
$\mv{H}=\mv{Q}\mv{\Lambda}^{1/2}\mv{U}^{\dag}$ where
$\mv{Q}\in\mathbb{C}^{M_{r,s}\times M_s}$ and
$\mv{U}\in\mathbb{C}^{M_{t,s}\times M_s}$ are matrices with
othornormal columns, $M_s=\min(M_{t,s},M_{r,s})$, and $\mv{\Lambda}$
is a $M_s \times M_s$ diagonal and positive matrix for which its
diagonal elements are denoted by $\lambda_1\geq\lambda_2 \geq \ldots
\geq \lambda_{M_s}$. Let $\mv{V}=\mv{U}$, and furthermore,
$\tilde{\mv{y}}(n)=\mv{Q}^{\dag}\mv{y}(n)$ and
$\mv{x}(n)=\mv{U}\tilde{\mv{x}}(n)$, the secondary user's MIMO
channel in (\ref{eq:channel model}) can be equivalently written as
\begin{eqnarray}
\tilde{\mv{y}}(n)&=&
\mv{Q}^{\dag}\mv{H}\mv{U}\tilde{\mv{x}}(n)+\tilde{\mv{z}}(n),
\\ &=&\mv{\Lambda}^{1/2}\tilde{\mv{x}}(n)+\tilde{\mv{z}}(n),\label{eq:equiv channel}
\end{eqnarray}
where $\tilde{\mv{z}}(n)=\mv{Q}^{\dag}\mv{z}(n)$  and
$\tilde{\mv{z}}(n)\sim\mathcal{CN}(0,\mv{I})$. Notice that the D-SVD
diagonalizes the MIMO channel and decomposes it into $M_s$ parallel
sub-channels with channel gains $\sqrt{\lambda_i}, i=1,\ldots,M_s$.

Let $\mv{U}\triangleq[\mv{u}_1, \mv{u}_2, \ldots, \mv{u}_{M_s}]$,
and $\alpha_i=||\mv{g}\mv{u}_{i}||^2, i=1,\ldots,M_s$. Considering
the equivalent channel in (\ref{eq:equiv channel}) and noticing that
$\mathbb{E}[\tilde{\mv{x}}(n)\tilde{\mv{x}}^{\dag}(n)]=\mv{\Sigma}$,
the optimal power assignments $\sigma_i$'s can be obtained by
solving the following equivalent problem ({\bf P3}) derived from
Problem {\bf P2}:
\begin{eqnarray}
\mathtt{Maximize}&& \sum_{i=1}^{M_s}\log_2\left(1+\lambda_i
\sigma_i\right)
\\ \mathtt {Subject \ to} && \sum_{i=1}^{M_s}\sigma_i\leq P_t, \label{eq:Tx power SVD}\\ &&
\sum_{i=1}^{M_s} \alpha_i\sigma_i\leq \gamma, \label{eq:Ic power SVD}\\
&& \sigma_i\geq 0, \forall i.
\end{eqnarray}
The optimal $\sigma_i$'s for the above problem can be shown to have
the following form of {\it multi-level} WF solutions:
\begin{equation}\label{eq:modified WF}
\sigma_i=\left(\frac{1}{\nu+\alpha_i\mu}-\frac{1}{\lambda_i}\right)^+,
\ i=1,\ldots,M_s,
\end{equation}
where $\nu$ and $\mu$ can be shown to be the non-negative Lagrange
multipliers \cite{Boydbook} associated with the transmit-power
constraint in (\ref{eq:Tx power SVD}) and the interference-power
constraint in (\ref{eq:Ic power SVD}), respectively, and can be
obtained by the following algorithm ({\bf A1}):
\begin{itemize}
\item {\bf Given} $\mu\in [0, \hat{\mu}]$
\item {\bf Initialize} $\mu_{\min}=0$, $\mu_{\max}=\hat{\mu}$
\item {\bf Repeat}
\begin{itemize}
\item[1.] Set $\mu\leftarrow\frac{1}{2}(\mu_{\min}+\mu_{\max})$.
\item[2.] Find the minimum $\nu$, $\nu\geq 0$, with which $\sum_{i=1}^{M_s}(\frac{1}{\nu+\alpha_i\mu}-\frac{1}{\lambda_i})^+\leq
P_t$. Substitute the obtained $\nu$ into (\ref{eq:modified WF}) to
obtain $\sigma_i$'s.
\item[3.] Update $\mu$ by the bisection method \cite{Boydbook}: If
$\sum_{i=1}^{M_s}\alpha_i\sigma_i\geq \gamma$, set
$\mu_{\min}\leftarrow\mu$; otherwise, $\mu_{\max}\leftarrow\mu$.
\end{itemize}
\item {\bf Until} $\mu_{\max}-\mu_{\min}\leq\delta_{\mu}$ where $\delta_{\mu}$ is a small positive constant
that controls the algorithm accuracy.
\end{itemize}

\begin{remark}\label{remark:low-SNR}
If the interference-power constraint in (\ref{eq:Ic power SVD}) is
inactive, it follows from the Karush-Kuhn-Tacker (KKT) conditions
\cite{Boydbook} that $\mu=0$. From (\ref{eq:modified WF}), the
allocated powers can be then written as
$\sigma_i=\left(\frac{1}{\nu}-\frac{1}{\lambda_i}\right)^+, \
i=1,\ldots,M_s$, which become equal to the standard WF solutions
with a constant water-level $\frac{1}{\nu}$. Hence, the D-SVD is
indeed optimal if $\gamma$ is sufficiently large such that the
interference-power constraint is inactive.
\end{remark}

\begin{remark}
If the secondary user's channel is MISO, it is not hard to verify
that the D-SVD algorithm results in $\mv{S}=\mv{v}\mv{v}^{\dag}$,
where
$\mv{v}=\frac{\sqrt{\min(\gamma,P_t)}}{||\mv{h}||}\mv{h}^{\dag}$.
Comparing it with the optimal $\mv{S}$ by Theorem
\ref{theorem:MISO}, it follows that the D-SVD is optimal only if
Case I of Theorem \ref{theorem:MISO} is true, i.e., $\gamma$ is
sufficiently large such that the interference-power constraint is
inactive.
\end{remark}

\subsubsection{Projected-Channel SVD (P-SVD)}
In the P-SVD, $\mv{S}$ is designed based on the so-called
zero-forcing (ZF) criterion so as to completely avoid any
interference at the primary receiver. It is noted that the ZF
criterion has also been used in the MIMO broadcast channel (MIMO-BC)
for design of downlink precoding that removes any co-channel
interference between users (e.g., \cite{Spencer04}-\cite{Tejera05}).
For the P-SVD, the secondary user's channel $\mv{H}$ is first
projected into the null space of $\mv{g}^{\dag}$ as
\begin{equation}\label{eq:channel projection}
\mv{H}_{\bot}=\mv{H}\left(\mv{I}-\hat{\mv{g}}\hat{\mv{g}}^{\dag}\right).
\end{equation}
Let the SVD of the projected channel $\mv{H}_{\bot}$ be expressed as
$\mv{H}_{\bot}=\mv{Q}_{\bot}\mv{\Lambda}_{\bot}(\mv{U}_{\bot})^{\dag}$.
Then, the precoding matrix $\mv{V}$ for $\mv{S}$ is taken as
$\mv{V}=\mv{U}_{\bot}$. From (\ref{eq:channel projection}), by
multiplying both the left-hand-side (LHS) and right-hand-side (RHS)
with $\hat{\mv{g}}$, it can be verified that
$(\mv{U}_{\bot})^{\dag}\hat{\mv{g}}=0$. Since
$\mv{S}=\mv{U}_{\bot}\mv{\Sigma}(\mv{U}_{\bot})^{\dag}$, it then
follows that the resultant interference power at the primary
receiver is zero, i.e., $\mv{g}\mv{S}\mv{g}^{\dag}=0$, and hence the
interference-power constraint (\ref{eq:Ic power}) for Problem {\bf
P2} is always satisfied for any $\gamma$, $\gamma\geq 0$.

Like the D-SVD, the P-SVD also diagonalizes the secondary user's
MIMO channel. To see it, let
$\bar{\mv{y}}(n)=(\mv{Q}_{\bot})^{\dag}\mv{y}(n)$ and
$\mv{x}(n)=\mv{U}_{\bot}\bar{\mv{x}}(n)$, and the secondary channel
in (\ref{eq:channel model}) can be equivalently represented by
\begin{eqnarray}
\bar{\mv{y}}(n)&=&
(\mv{Q}_{\bot})^{\dag}\mv{H}\mv{U}_{\bot}\bar{\mv{x}}(n)+\bar{\mv{z}}(n),
\\ &=& (\mv{Q}_{\bot})^{\dag}(\mv{H}_{\bot}+
\mv{H}\hat{\mv{g}}\hat{\mv{g}}^{\dag})\mv{U}_{\bot}\bar{\mv{x}}(n)+\bar{\mv{z}}(n),
\label{eq:derivation 1}
\\ &=&
(\mv{Q}_{\bot})^{\dag}\mv{H}_{\bot}\mv{U}_{\bot}\bar{\mv{x}}(n)+\bar{\mv{z}}(n),
\label{eq:derivation 2}
\\ &=&(\mv{\Lambda}_{\bot})^{1/2}\bar{\mv{x}}(n)+\bar{\mv{z}}(n) \label{eq:P-SVD channel},
\end{eqnarray}
where $\bar{\mv{z}}(n)=(\mv{Q}_{\bot})^{\dag}\mv{z}(n)$  and
$\bar{\mv{z}}(n)\sim\mathcal{CN}(0,\mv{I})$, (\ref{eq:derivation 1})
is from (\ref{eq:channel projection}), and (\ref{eq:derivation 2})
is by $\hat{\mv{g}}^{\dag}\mv{U}_{\bot}=0$. Because the projection
in (\ref{eq:channel projection}) reduces the channel rank at most by
one, $\mv{\Lambda}_{\bot}$ has $M^{\bot}_s=\min(M_{t,s}-1,M_{r,s})$
diagonal elements, denoted by
$\lambda^{\bot}_1\geq\ldots\geq\lambda^{\bot}_{M^{\bot}_s}$. Hence,
the secondary MIMO channel is decomposed into $M^{\bot}_s$
sub-channels with channel gains $\sqrt{\lambda^{\bot}_i}, i=1,\ldots
M^{\bot}_s$. The power allocation $\mv{\Sigma}$ for these
sub-channels can be then obtained by considering the equivalent
channel (\ref{eq:P-SVD channel}) with
$\mathbb{E}[\bar{\mv{x}}(n)\bar{\mv{x}}^{\dag}(n)]=\mv{\Sigma}$, as
the standard WF solutions
$\sigma_i=\left(\nu'-\frac{1}{\lambda_i^{\bot}}\right)^+, \
i=1,\ldots,M^{\bot}_s$, where $\nu'$ is the constant water-level
such that $\sum_{i=1}^{M^{\bot}_s}\sigma_i=P_t$.

\begin{remark}
If the interference-power constraint for Problem {\bf P2} is indeed
$\gamma=0$, it is conjectured that the P-SVD is optimal. This
conjecture is proved in Appendix \ref{appendix:optimality of P-SVD}.
\end{remark}

\begin{remark}
If the secondary user's channel is MISO, it is not hard to verify
that the P-SVD algorithm results in $\mv{S}=\mv{v}\mv{v}^{\dag}$,
where $\mv{v}= \sqrt{P_t}\frac{\beta_h}{||\beta_h||}
\hat{\mv{h}}_{\bot}$. Comparing it with the optimal $\mv{S}$ by
Theorem \ref{theorem:MISO}, it follows that the P-SVD is in general
suboptimal unless $\gamma=0$ and then Case II of Theorem
\ref{theorem:MISO} applies.
\end{remark}

\subsubsection{Performance Comparison} \label{subsubsec:performance
comparison}

Comparing the computational complexity of the D-SVD and P-SVD, it is
noticed that the former has a larger complexity for determining
optimal power allocations due to the multi-level WF, while so does
the latter for obtaining the precoding matrix due to the additional
channel projection. In the following, the maximum achievable rate
(or the capacity) for the secondary link by the optimal $\mv{S}$,
and that by the two SVD-based algorithms are compared under two
extreme scenarios: $P_t \rightarrow 0$ and $P_t \rightarrow \infty$,
both under the assumption that $\gamma$ is finite and $\gamma>0$.
Notice that the former case can also be considered as asymptotically
low signal-to-noise ratio (SNR) at the secondary receiver (if the
transmit power $P$ takes its maximum value $P_t$) while the latter
case as asymptotically high SNR.

As $P_t \rightarrow 0$, the D-SVD is the optimal solution. This can
be easily verified similar like in Remark \ref{remark:low-SNR} by
observing that the interference-power constraint in (\ref{eq:Ic
power}) eventually becomes inactive as $P_t\rightarrow 0$. In
contrast, the P-SVD may incur a non-negligible rate loss in this
case. This is so because the ZF-based projection causes part of the
secondary user's channel energy to be lost. Since as $P_t\rightarrow
0$, both SVD-based algorithms assign the total transmit power to the
sub-channel with the largest channel gain ($\lambda_1$ for the D-SVD
and $\lambda^{\bot}_1$ for the P-SVD), it can be easily verified
that the achievable rates of the D-SVD and P-SVD become
$\log_2(1+\lambda_1P_t)$ and $\log_2(1+\lambda^{\bot}_1P_t)$,
respectively. Notice that $\lambda_1\geq \lambda^{\bot}_1$. Because
$\log(1+x)\cong x$ as $x\rightarrow 0$, the following theorem is
obtained for the achievable rate by the optimal $\mv{S}$, the D-SVD
and the P-SVD, denoted by $R_{\rm opt}$, $R_{\rm D-SVD}$, and
$R_{\rm P-SVD}$, respectively:
\begin{theorem}\label{theorem:low P}
As $P_t \rightarrow 0$,  $R_{\rm opt}=R_{\rm D-SVD}=
\frac{\lambda_1P_t}{\log2}$, and $R_{\rm
P-SVD}=\frac{\lambda_1^{\bot}P_t}{\log2}$.
\end{theorem}

On the other hand, as $P_t \rightarrow \infty$, the achievable rate
by the P-SVD becomes close to the secondary channel capacity
achievable by the optimal $\mv{S}$. This is so because as $P_t$
increases and eventually exceeds some certain threshold, any
additional transmit power needs to be allocated into the projected
channel $\mv{H}^{\bot}$ in order not to violate the
interference-power constraint. From (\ref{eq:P-SVD channel}), it can
be easily verified that as $P_t \rightarrow \infty$, the achievable
rate of the P-SVD has a linear increase with $\log_2 P_t$ by a
factor of $M^{\bot}_s=\min(M_{r,s},M_{t,s}-1)$, which is usually
termed as the pre-log factor, the degree of freedom, or the spatial
multiplexing gain in the literature. The following lemma states that
$R_{\rm opt}$ also has a spatial multiplexing gain of $M^{\bot}_s$,
indicating that the P-SVD is indeed optimal (in terms of the
achievable spatial multiplexing gain) as $P_t\rightarrow \infty$.
\begin{lemma}\label{lemma:multiplexing gain}
As $P_t \rightarrow \infty$, $\frac{R_{\rm
opt}}{\log_2P_t}=\min(M_{r,s},M_{t,s}-1)$.
\end{lemma}
\begin{proof}
Please refer to Appendix \ref{appendix:proof of lemma 3}.
\end{proof}
In contrast, the D-SVD may suffer a significant rate loss when
$P_t\rightarrow \infty$. This can be observed from (\ref{eq:Ic power
SVD}). Suppose $\alpha_i>0, \forall i$, the total transmit power $P$
is upper-bounded as $P=\sum_i \sigma_i\leq \frac{\gamma}{\min_i
\alpha_i}$, regardless of the secondary user's own transmit power
constraint $P_t$. As a result, the achievable rate of the D-SVD
eventually gets saturated as $P_t\rightarrow \infty$, and the
corresponding spatial multiplexing gain can be easily shown equal to
zero. The above results are summarized in the following theorem:
\begin{theorem}\label{theorem:high P}
As $P_t \rightarrow \infty$,  $\frac{R_{\rm
opt}}{\log_2P_t}=\frac{R_{\rm
P-SVD}}{\log_2P_t}=\min(M_{r,s},M_{t,s}-1)$, and $\frac{R_{\rm
D-SVD}}{\log_2P_t}=0$ if $\alpha_i>0, i=1,\ldots,M_s$.
\end{theorem}

\section{Multiple Primary Receivers}\label{sec:multiple primary receivers}

In this section, the general case of multiple primary receivers each
having single or multiple receive antennas as shown in Fig.
\ref{fig:CR system} is studied. Hence, Problem {\bf P1} in
(\ref{eq:capacity P1})-(\ref{eq:S P1}) is considered.

\subsection{MISO Secondary User's Channel}
First, we consider the special case where the secondary user's
channel is MISO. Like Lemma \ref{lemma:MISO 1} in the case of one
single-antenna primary receiver, it can be shown that the optimal
transmit covariance matrix $\mv{S}$ in the case of multiple primary
receivers/antennas still remains as a rank-one matrix when the
secondary user's channel is MISO, i.e., beamforming is optimal.
Hence, it follows that $\mv{S}=\mv{v}\mv{v}^{\dag}$. However, a
closed-form solution for the optimal beamforming vector $\mv{v}$
seems unlikely when the total number of antennas from all primary
receivers $M_{r,p}=\sum_{k=1}^{M_k}$ is greater than one.
Nevertheless, the optimal rank-one $\mv{S}$ for Problem {\bf P1} can
still be obtained using numerical optimization techniques like the
interior-point method. Alternatively, the two SVD-based algorithms
can also be modified to obtain suboptimal $\mv{v}$, as will be
illustrated later in Section \ref{subsec:MIMO case multiple primary
receivers} when the more general case of MIMO secondary user's
channel is addressed.

In this subsection, we first consider Problem {\bf P1} by
substituting $\mv{S}=\mv{v}\mv{v}^{\dag}$ into the problem. Then,
based on this reformulated problem, we present alternative numerical
optimization techniques that, in most cases, can be more efficient
in terms of computational complexity than the interior-point method.
The reformulated problem of Problem {\bf P1} when the secondary
user's channel is MISO, i.e., $\mv{H}\equiv\mv{h}$, can be expressed
as ({\bf P4})
\begin{eqnarray}
\mathtt{Maximize}&& \left\|\mv{h}\mv{v}\right\|^2 \label{eq:capacity
P4}
\\ \mathtt {Subject \ to} && \|\mv{v}\|^2\leq P_t, \label{eq:Tx power P4}\\  && \left\|\mv{G}_k\mv{v}\right\|^2\leq
\Gamma_k, \ k=1,\ldots,K. \label{eq:Ic power P4}
\end{eqnarray}
Notice that for the above problem although the two constraints
specify a convex set for $\mv{v}$, the objective function is
non-concave of $\mv{v}$ and hence renders Problem {\bf P4}
non-convex in its direct form. However, an interesting observation
here is that if $\mv{v}$ satisfies the two constraints (\ref{eq:Tx
power P4}) and (\ref{eq:Ic power P4}), so does $e^{j\theta}\mv{v}$
for any arbitrary $\theta$, and the value of the objective function
is maintained. Thus, without loss of generality we can assume that
$\mv{h}\mv{v}$ is a real number. Using this assumption, Problem {\bf
P4} can be rewritten as
\begin{eqnarray}
\mathtt{Maximize}&& \mathtt{Re}(\mv{h}\mv{v})
\\ \mathtt {Subject \ to} && \mathtt{Im}(\mv{h}\mv{v})=0, \\ && \|\mv{v}\|^2\leq P_t,\\  && \left\|\mv{G}_k\mv{v}\right\|^2\leq
\Gamma_k, \ k=1,\ldots,K.
\end{eqnarray}
It can be shown that the above problem can be cast as a second-order
cone programming (SOCP) \cite{Boydbook}, which can be solved by
standard numerical optimization software.

Since Problem {\bf P4} is indeed convex, and it can be further
verified that the Slater's condition \cite{Boydbook} that requires
that the original (primal) problem needs to have a non-empty
interior of the feasible set is satisfied for this problem, it
follows that the strong duality holds for the problem at hand, which
ensures that the duality gap between the original problem and its
Lagrange dual problem is zero, i.e., solving the Lagrange dual
problem is equivalent to solving the original problem. Motivated by
this fact, as follows we will present an alternative solution for
Problem {\bf P4} by considering its Lagrange dual problem. First,
the Lagrangian \cite{Boydbook} of Problem {\bf P4} can be expressed
as
\begin{eqnarray}\label{eq:Lagrange MISO}
\mathcal{L}(\mv{v},\nu,\{\mu_k\})=\mv{v}^{\dag}\mv{h}^{\dag}\mv{h}\mv{v}-\nu\left(\mv{v}^{\dag}\mv{v}-P_t\right)
-\sum_{k=1}^K\mu_k\left(\mv{v}^{\dag}\mv{G}_k^{\dag}\mv{G}_k\mv{v}-\Gamma_k\right),
\end{eqnarray}
where $\nu$ and $\mu_k$, $k=1,\ldots,K$, are non-negative Lagrange
multiplier (dual variable) associated with the transmit-power
constraint in (\ref{eq:Tx power P4}), and each interference-power
constraint in (\ref{eq:Ic power P4}), respectively. Let
$\mathcal{D}_{v}$ be the set specified by the constraints
(\ref{eq:Tx power P4}) and (\ref{eq:Ic power P4}), the Lagrange dual
function \cite{Boydbook} can be then defined as
\begin{eqnarray}\label{eq:Lagrange dual MISO}
f(\nu,\{\mu_k\})=\max_{\mv{v}\in\mathcal{D}_{v}}
\mathcal{L}(\mv{v},\nu,\{\mu_k\}).
\end{eqnarray}
The Lagrangian dual problem \cite{Boydbook} is then defined as
\begin{equation}\label{eq:dual problem MISO}
\min_{\nu\geq 0,\mu_k\geq 0, \forall k}f(\nu,\{\mu_k\}).
\end{equation}
Notice that the Lagrangian in (\ref{eq:Lagrange MISO}) can be
rewritten as
\begin{eqnarray}\label{eq:Lagrange MISO 2}
\mv{v}^{\dag}\left(\mv{h}^{\dag}\mv{h}-\nu\mv{I}-\sum_{k=1}^K\mu_k\mv{G}_k^{\dag}\mv{G}_k\right)\mv{v}
+\nu P_t+\sum_{k=1}^K\mu_k\Gamma_k.
\end{eqnarray}
From the above expression, it follows that for solving the
maximization problem so as to obtain the dual function in
(\ref{eq:Lagrange dual MISO}), the dual variables $\nu$ and
$\{\mu_k\}$ must satisfy
\begin{equation}\label{eq:dual condition}
\mv{h}^{\dag}\mv{h}-\nu\mv{I}-\sum_{k=1}^K\mu_k\mv{G}_k^{\dag}\mv{G}_k\preceq
0,
\end{equation}
because otherwise, the optimal $\mv{v}$ that maximizes the
Lagrangian in (\ref{eq:Lagrange MISO 2}) will go unbounded to
infinity, which contradicts the facts that the optimal value of the
primal problem is bounded, and the duality gap is zero. Using
(\ref{eq:dual condition}), the dual problem in (\ref{eq:dual problem
MISO}) can be expressed as ({\bf P5})
\begin{eqnarray}
\mathtt{Minimize}&& \nu P_t+\sum_{k=1}^K\mu_k\Gamma_k
\\ \mathtt {Subject \ to} && -\mv{h}^{\dag}\mv{h}+\nu\mv{I}+\sum_{k=1}^K\mu_k\mv{G}_k^{\dag}\mv{G}_k\succeq 0, \\ && \nu\geq 0, \\
&&  \mu_k\geq 0, \ k=1,\ldots,K.
\end{eqnarray}
The above problem can be efficiently solved using the standard
semi-definite programming (SDP) software \cite{Boydbook}. From the
dual optimal solutions, the primal optimal solution for $\mv{v}$ can
be obtained as any vector (may not be unique) that satisfies
$\left(\mv{h}^{\dag}\mv{h}-\nu\mv{I}-\sum_{k=1}^K\mu_k\mv{G}_k^{\dag}\mv{G}_k\right)\mv{v}=0$.

The computational complexity of three different methods for solving
Problem {\bf P1}, namely, the interior-point method, the SOCP for
{\bf P4}, and the SDP for {\bf P5}, in the case of MISO secondary
user's channel are compared as follows. All three methods have the
similar order of complexity in terms of the number of (real)
variables in each corresponding problem. For the interior-point
method, the unknown $\mv{S}$ has $M_{t,s}$ real variables and
$\frac{(M_{t,s}-1)M_{t,s}}{2}$ complex variables, results in a total
number of real variables equal to $M_{t,s}^2$. The SOCP has the
unknown $\mv{v}$ that has $2M_{t,s}$ real variables. The SDP has the
unknown dual variables $\nu$ and $\{\mu_k\}$ that in total
contribute to $K+1$ real variables. Consequently, the SOCP has a
much lower computational complexity than the interior-point method.
So does the SDP than the SOCP when, e.g., $K\ll M_{t,s}$. At last,
it is noted that the SOCP and SDP have also been applied for design
of downlink precoding for MIMO-BC channels (e.g.,
\cite{Bengtsson99}-\cite{Wiessel04}).

\subsection{MIMO Secondary User's Channel}\label{subsec:MIMO case multiple primary receivers}

When the secondary user's channel is MIMO, in general we have to
resort to the interior-point method for solving Problem {\bf P1}.
Alternatively, the two SVD-based solutions developed in Section
\ref{sec:single primary receiver} in the case of one single-antenna
primary receiver can also be modified to incorporate more than one
primary receivers/antennas, as will be shown in this subsection.

\subsubsection{Direct-Channel SVD (D-SVD)}
For the D-SVD, the precoding matrix $\mv{V}$ remains to be $\mv{U}$
obtained from the SVD of $\mv{H}$ regardless of the number of
primary receivers or antennas. However, the power allocation
$\mv{\Sigma}$ needs to be adjusted in order to incorporate multiple
interference-power constraints in (\ref{eq:total power ct}). By
introducing a set of non-negative Lagrange multipliers,
$\mu_1,\ldots, \mu_K$, each associated with one of $K$
interference-power constraints in (\ref{eq:total power ct}). Like
Problem {\bf P3} in the case of one single-antenna primary receiver,
the optimal power assignments $\sigma_i$'s in the case of multiple
primary receivers/antennas can also be obtained as the {\it
multi-level} WF solutions:
\begin{equation}
\sigma_i=\left(\frac{1}{\nu+\sum_{k=1}^{K}\sum_{j=1}^{M_k}\alpha_{i,k,j}\mu_k}-\frac{1}{\lambda_i}\right)^+,
\ i=1,\ldots,M_s,
\end{equation}
where $\alpha_{i,k,j}=||\mv{g}_{k,j}\mv{u}_i||^2$. Algorithm {\bf
A1} also needs to be modified so as to iteratively search for the
optimal $\nu$ and $\mu_k$'s. Instead of updating a single $\mu$ by
the bisection method in Algorithm {\bf A1}, $\mu_k$'s can be
simultaneously updated by the ellipsoid method \cite{BGT81} by
observing that
$\gamma_k-\sum_{i=1}^{M_s}\sum_{j=1}^{M_k}\alpha_{i,j,k}\sigma_i$ is
a {\it sub-gradient} of $\mu_k, k=1,\ldots,K$.

\subsubsection{Projected-Channel SVD (P-SVD)} \label{subsubsection:MU
P-SVD} Let $\mv{G} \in \mathbb{C}^{M_{r,p}\times M_{t,s}}$ denote
the channel from the secondary transmitter to all primary
receivers/antennas by taking each $\mv{g}_{k,j}$ as the
$(\sum_{k'=1}^{k}M_{k'-1}+j)$-th row of $\mv{G}$, where $M_0=0$.
Alternatively, $\mv{G}=[\mv{G}_1^T,\ldots,\mv{G}_K^T]^T$. Let the
SVD of $\mv{G}$ be expressed as
$\mv{G}=\mv{Q}_{G}\mv{\Lambda}_{G}^{1/2}\mv{U}_{G}^{\dag}$. In order
to avoid completely the interference at all primary
receivers/antennas, the P-SVD obtains the precoding matrix $\mv{V}$
from the SVD of $\mv{H}_{\bot}$, which is the projection of $\mv{H}$
into the null space of $\mv{G}^{\dag}$, i.e.,
\begin{equation}\label{eq:channel projection multiple primary
receivers}
\mv{H}_{\bot}=\mv{H}(\mv{I}-\mv{U}_G\mv{U}_G^{\dag}).
\end{equation}
Notice that the above matrix projection is non-trivial only when
$M_{t,s}> M_{r,p}$, because otherwise a null matrix $\mv{H}_{\bot}$
is resulted. Let the SVD of $\mv{H}_{\bot}$ be expressed as
$\mv{H}_{\bot}=\mv{Q}_{\bot}(\mv{\Lambda}_{\bot})^{1/2}(\mv{U}_{\bot})^{\dag}$.
Like the case of a single-antenna primary receiver, it can be
verified that the secondary user's MIMO channel can be decomposed
into $\min(M_{r,s}, M_{t,s}-M_{r,p})$ sub-channels with channel
gains given by the diagonal elements of $\mv{\Lambda}_{\bot}$.

\subsubsection{Hybrid D-SVD/P-SVD} \label{subsubsection:Hybrid SVD}

A new hybrid D-SVD/P-SVD algorithm is proposed here in the case of
multiple primary receivers/antennas, i.e., $M_{r,p}> 1$. In this
algorithm, the secondary user's MIMO channel is first projected into
the null space of some selected subspace of $\mv{G}^{\dag}$ as
opposed to the whole space spanned by $\mv{G}^{\dag}$ in the P-SVD.
Then, the D-SVD algorithm is applied to this projected channel to
obtain transmit precoding matrix as well as power allocations that
satisfy the interference-power constraints at the primary receivers.
There are a couple of reasons for using this hybrid D-SVD/P-SVD
algorithm. First, it works even if $M_{t,s}\leq M_{r,p}$ under which
the P-SVD is not implementable. Secondly, in the case of a
single-antenna primary receiver, i.e., $M_{r,p}=1$, in Theorem
\ref{theorem:low P} and Theorem \ref{theorem:high P}, it has be
shown that the D-SVD and P-SVD are asymptotically optimal as
$P_t\rightarrow 0$ and $P_t\rightarrow \infty$, respectively. It is
thus conjectured that when $M_{r,p}> 1$, the hybrid algorithm is
likely to perform superior than both the D-SVD and P-SVD for some
moderate values of $P_t$ (as will be verified later by the
simulation results in Section \ref{sec:simulation results}).

One important issue to address for the hybrid D-SVD/P-SVD algorithm
is on the selection for the projected subspace of $\mv{G}^{\dag}$.
Although the optimal selection rule remains unknown, in this paper a
heuristic rule is proposed as follows. First, each channel
$\mv{g}_{k,j}$ in $\mv{G}$ is normalized by the corresponding
$\sqrt{\Gamma_k}$ so as to make the equivalent interference-power
constraints at all primary receivers equal to unity. Denote this
normalized $\mv{G}$ as $\hat{\mv{G}}$. Next, let the SVD of
$\hat{\mv{G}}$ be expressed as
$\hat{\mv{G}}=\mv{Q}_{\hat{G}}\mv{\Lambda}_{\hat{G}}^{1/2}\mv{U}_{\hat{G}}^{\dag}$,
where the singular values in $\mv{\Lambda}_{\hat{G}}^{1/2}$ are
ordered by $\lambda_{\hat{G},1}\geq \ldots \geq \lambda_{\hat{G},
M_{\hat{G}}}, M_{\hat{G}}=\min(M_{t,s},M_{r,p})$. The hybrid
algorithm then projects $\mv{H}$ into the null space of the space
spanned by the first $b$ , $b\leq M_{\hat{G}}$, column vectors of
$\mv{U}_{\hat{G}}$, denoted by $\mv{U}_{\hat{G}}(b)$, corresponding
to the first $b$ largest singular values of $\hat{\mv{G}}$, i.e.,
\begin{equation}\label{eq:channel projection hybrid}
\mv{H}_{\bot}(b)=\mv{H}\left(\mv{I}-\mv{U}_{\hat{G}}(b)\left(\mv{U}_{\hat{G}}(b)\right)^{\dag}\right).
\end{equation}
After applying the precoding matrix $\mv{V}$ based on the SVD of
$\mv{H}_{\bot}(b)$, it can be verified that the secondary user's
MIMO channel is diagonalized and decomposed into
$\min(M_{r,s},M_{t,s}-b)$ sub-channels. Notice that if
$b<M_{\hat{G}}$, there may be remaining interference at each primary
receiver and, hence, the secondary transmitter power allocation
$\mv{\Sigma}$ needs to be adjusted to make the interference-power
constraints being satisfied at all primary receivers. This can be
done by solving a similar problem like Problem {\bf P3} in the D-SVD
case.

The rational for the hybrid D-SVD/P-SVD algorithm lies in that only
some selected subspace of $\hat{\mv{G}}^{\dag}$ at the primary
receivers is removed by the secondary transmit precoding, while the
remaining subspace of $\hat{\mv{G}}^{\dag}$ is preserved for the
secondary transmission. Notice that the nulled subspace on the
average contributes to the most amount of interference powers at the
primary receivers because it corresponds to the first $b$ largest
singular values of $\hat{\mv{G}}$. For typical wireless channels,
the channel from the secondary transmitter to each primary receive
antenna (represented by the corresponding row vector of $\mv{G}$) is
usually subject to variation or fading over time. Furthermore, these
channels may exhibit quite different average channel gains (e.g., a
near-far situation), and some of them may have certain correlations
(e.g., for the receive antennas at the same primary receiver or for
the primary receives located in the same vicinity). All these
factors can result in a more spread distribution for the singular
values in $\mv{\Lambda}_{\hat{G}}^{1/2}$, which thereby makes the
subspace projection more effective in removing the interferences at
the primary receivers. Because $\mv{G}$ is usually independent of
the secondary MIMO channel $\mv{H}$, it is unlikely that the nulled
subspace of $\hat{\mv{G}}^{\dag}$ are strongly correlated with
$\mv{H}$, a phenomenon termed {\it interference diversity}.
Therefore, by selecting a proper value of $b$, the secondary user
has a more flexibility in balancing between spacial multiplexing for
the secondary transmission and interference avoidance at the primary
receivers. Notice that the interference-diversity effect is in
principle very similar to the {\it multiuser diversity} effect
previously reported in the literature (e.g.,
\cite{Knopp95}-\cite{Tse02}), which is exploited for multiuser rate
scheduling in a mobile wireless network.

\section{Multi-Channel Transmission}\label{sec:multi-channel}

So far, we have considered the single-channel transmission for both
primary and secondary users, and shown that even when some primary
users are active for transmission, the secondary user is still able
to achieve opportunistic spectrum sharing with active primary users
by utilizing multiple transmit antennas and properly designing its
transmit spatial spectrum. In this section, a more general
multi-channel transmission is studied where both primary and
secondary users transmit over parallel single-channels. This
scenario is applicable when, e.g., both primary and secondary users
transmit over multi-tone or
orthogonal-frequency-division-multiplexing (OFDM) channels, or
alternatively, over consecutive block-fading channels. For such
scenarios, in order to achieve optimum spectrum sharing, the
secondary user needs to first detect the primary user's transmission
activities in all available dimensions of space, time and frequency
and then adapts its transmit resources such as power, rate, and
spatial spectrum at all of these dimensions. For convenience, in
this section we consider the case of one single-antenna primary
receiver in the CR network although the generalization to multiple
primary receivers with single or multiple antennas can also be done
similarly like in Section \ref{sec:multiple primary receivers}.

Let $\mv{H}_j$, $\mv{g}_j$ and $\mv{S}_j$ be defined the same as in
Section \ref{sec:single primary receiver} but now corresponding to
one particular sub-channel $j, j=1,\ldots,N$, where $N$ is the total
number of sub-channels for the multi-band transmission. Each
sub-channel can be considered as, e.g., one OFDM tone in frequency
domain or one fading block in time domain. It is assumed that the
interference-power constraint $\gamma$ is identical at all
sub-channels. Problem {\bf P2} can be then extended to the
multi-channel transmission as ({\bf P6})
\begin{eqnarray}
\mathtt{Maximize}&&
\sum_{j=1}^N\log_2\left|\mv{I}+\mv{H}_j\mv{S}_j\mv{H}_j^{\dag}\right|
\label{eq:capacity OFDM}
\\ \mathtt {Subject \ to} && \sum_{j=1}^{N}\mathtt{Tr}(\mv{S}_j)\leq P_t, \label{eq:Tx power OFDM}\\
&& \mv{g}_j\mv{S}_j\mv{g}_j^{\dag}\leq \gamma, \forall j,
\label{eq:Ic power OFDM} \\ && \mv{S}_j\succeq 0, \forall j.
\label{eq:S OFDM}
\end{eqnarray}
This problem maximizes the total transmit rate over all of $N$
sub-channels for the secondary transmission under a total
transmit-power constraint and a set of interference-power
constraints each for one of $N$ sub-channels. Like Problem {\bf P2},
the above problem can be shown to be convex and, hence, can be
solved using standard convex optimization techniques. In the
following, the {\it Lagrange dual-decomposition method} is applied
to decompose Problem {\bf P6} into $N$ subproblems all having an
identical structure. The main advantage of this method lies in that
if the same computational routine can be simultaneously applied to
all $N$ subproblems, the overall computational time is maintained
regardless of $N$. Other applications of the Lagrange
dual-decomposition method for resource allocation in communication
systems can be found in, e.g., \cite{Yu03}-\cite{Chen07}.

For the problem at hand, the first step is to introduce the
non-negative Lagrange multiplier $\nu$ associated with the power
constraint in (\ref{eq:Tx power OFDM}) and to write the Lagrangian
of the original (primal) problem as
\begin{eqnarray}\label{eq:Lagrange power min}
\mathcal{L}(\{\mv{S}_j\},\nu)=\sum_{j=1}^N\log_2\left|\mv{I}+\mv{H}_j\mv{S}_j\mv{H}_j^{\dag}\right|-
\nu\left(\sum_{j=1}^{N}\mathtt{Tr}(\mv{S}_j)-P_t\right).
\end{eqnarray}
Let $\mathcal{D}_j$ be the set specified by the remaining
constraints in (\ref{eq:Ic power OFDM}) and (\ref{eq:S OFDM})
corresponding to sub-channel $j, j=1,\ldots,N$. The Lagrange dual
function is then defined as
\begin{eqnarray}\label{eq:Lagrange dual power min}
f(\nu)=\max_{\mv{S}_j\in\mathcal{D}_j, \forall j}
\mathcal{L}(\{\mv{S}_j\},\nu).
\end{eqnarray}
The dual problem is then defined as $\min_{\nu\geq 0}f(\nu)$. It can
be verified that the strong duality holds for the problem at hand,
and hence the duality gap is zero. Therefore, the primal problem can
be solved by first maximizing the Lagrangian $\mathcal{L}$ to obtain
the dual function $f(\nu)$, and then minimizing $f(\nu)$ over all
non-negative values of $\nu$.

Consider first the problem for obtaining $f(\nu)$ with some given
$\nu$. It is interesting to observe that $f(\nu)$ can be also
written as
\begin{eqnarray}\label{eq:Lagrange dual rewrite}
f(\nu)= \sum_{j=1}^{N}f'_j(\nu)+\nu P_t,
\end{eqnarray}
where
\begin{eqnarray}\label{eq:Lagrange dual redefine}
f'_j(\nu)=\max_{\mv{S}_j\in\mathcal{D}_j}\log_2\left|\mv{I}+\mv{H}_j\mv{S}_j\mv{H}_j^{\dag}\right|-
\nu\mathtt{Tr}(\mv{S}_j), \ j=1,\ldots,N.
\end{eqnarray}
Hence, $f(\nu)$ can be obtained by solving $N$ independent
subproblems, each for $f'_j(\nu)$, $j=1,\ldots,N$. Each subproblem
has the same structure and hence can be solved by using the same
computational routine. This practice is usually referred to as the
{\it dual decomposition}. From (\ref{eq:Lagrange dual redefine}),
each of these subproblems can be written more explicitly as ({\bf
P7})
\begin{eqnarray}
\mathtt{Maximize}&&
\log_2\left|\mv{I}+\mv{H}_j\mv{S}_j\mv{H}^{\dag}_j\right|-\nu\mathtt{Tr}(\mv{S}_j)
\\ \mathtt {Subject \ to} && \mv{g}_j\mv{S}_j\mv{g}^{\dag}_j\leq
\gamma, \\ && \mv{S}_j\succeq 0.
\end{eqnarray}
Compared with Problem {\bf P2}, the above problem has the difference
in that the transmit-power constraint in (\ref{eq:Tx power}) is
removed and a new term $\nu\mathtt{Tr}(\mv{S}_j)$ is subtracted in
the objective function. This is a consequence of a total-power
constraint over all $N$ sub-channels in the multi-channel case
instead of the per-sub-channel-based power constraint in the
single-channel case. It can be verified that Problem {\bf P7} is
also convex and hence can be solved by using, e.g., the SOCP or SDP
if the secondary user's channel is MISO (beamforming is optimal), or
the interior-point method if the secondary user's channel is MIMO
(spatial multiplexing is optimal).

The dual variable $\nu$ can be considered as a common price applied
to all sub-channels for regulating their allocated transmit powers.
It can be verified that in Problem {\bf P7}, for each sub-channel, a
larger $\nu$ will result in a smaller power consumption
$\mathtt{Tr}(\mv{S}_j)$ and vice versa. Hence, the optimal $\nu$ can
be found by a bisection search by comparing the optimal sum-power
over all $N$ sub-channels for a given $\nu$ with the transmit-power
constraint $P_t$. In summary, the following algorithm ({\bf A2}) can
be used to solve Problem {\bf P6}:
\begin{itemize}
\item {\bf Given} $\nu\in [0, \hat{\nu}]$
\item {\bf Initialize} $\nu_{\min}=0$, $\nu_{\max}=\hat{\nu}$
\item {\bf Repeat}
\begin{itemize}
\item[1.] Set $\nu\leftarrow\frac{1}{2}(\nu_{\min}+\nu_{\max})$.
\item[2.] Solve Problem {\bf P7} for $j=1,\ldots,N$ independently to obtain
$\{\mv{S}_j\}$.
\item[3.] If $\sum_{j=1}^{N}\mathtt{Tr}(\mv{S}_j)\leq P_t$, set $\nu_{\max}\leftarrow\nu$;
otherwise, $\nu_{\min}\leftarrow\nu$.
\end{itemize}
\item {\bf Until} $\nu_{\max}-\nu_{\min}\leq\delta_{\nu}$ where $\delta_{\nu}$ is a small positive constant
that controls the algorithm accuracy.
\end{itemize}

Algorithm {\bf A2} can also be applied when the D-SVD or P-SVD
algorithm is used to approximately solve each subproblem {\bf P7}.
The selection of the D-SVD or P-SVD at each sub-channel depends on
the dual variable $\nu$ and the channels $\mv{H}_j$ and $\mv{g}_j$.
Because $\mv{H}_j$ and $\mv{g}_j$ may change over different
sub-channels, in order to exploit this diversity, the best selected
SVD-based algorithm (the D-SVD or P-SVD) may also change from one
sub-channel to another.

\section{Simulation Results}\label{sec:simulation results}
For the simulations, it is assumed that in the CR network both the
channels from the secondary transmitter to the primary receivers
$\mv{G}$, and to the secondary receiver $\mv{H}$, are drawn from the
family of random vectors/matrices for which each element is
generated by independent CSCG variable distributed as
$\mathcal{CN}(0,0.1)$ for $\mv{G}$, and $\mathcal{CN}(0,1)$ for
$\mv{H}$, respectively. The secondary user's transmit-power
constraint $P_t$ sweeps from 1 to 100, which is equivalent to a
range of average SNRs per receive antenna (measured when the actual
transmit power $P$ is equal to $P_t$) from 0 dB to 20 dB. Monte
Carlo simulations with 1000 randomly generated channel-pairs
$(\mv{G},\mv{H})$ are implemented, and the average achievable rates
are plotted versus the SNR values. The results are shown for the
following cases:

\subsection{Capacity With (w/) versus Without (w/o) Interference-Power Constraint}
Fig.\ref{fig:capacity comp SU} compares the capacity for the
secondary user's transmission w/ and w/o the interference-power
constraint at a single-antenna primary receiver. In the case that
the interference-power constraint is applied, $\gamma$ is equal to
0.01. For the secondary receiver, it is assumed that $M_{r,s}=1$,
and two antenna configurations are considered for the secondary
transmitter: (1) Single-input single-output (SISO) case:
$M_{t,s}=1$; (2) MISO case: $M_{t,s}=4$. For both SISO and MISO
secondary user's channels, it is observed that the
interference-power constraint reduces the capacity of the secondary
transmission, especially at the high-SNR regime. However, the
capacity improvement by adding more transmit antennas at the
secondary transmitter is observed to be substantial. This is so
because in the SISO case, as SNR increases, the capacity is
eventually limited by the interference-power constraint instead of
the secondary user's own transmit power constraint; while in the
MISO case, the secondary transmitter is able to apply transmit
beamforming so as to avoid the interference at the primary receiver
and at the same time, guarantee its own rate increase with SNR.
Therefore, in addition to antenna array and diversity gains in the
conventional MISO point-to-point transmission, multi-antennas at the
secondary transmitter play another essential role in a CR network,
i.e., adapting transmit spatial spectrum away from the interfering
direction to the primary receiver, so as to achieve a high spectral
efficiency for the secondary transmission.

\subsection{D-SVD versus P-SVD}
Fig. \ref{fig:SVD comp SU} shows the achievable rates of two
SVD-based algorithms, the D-SVD and P-SVD, and compares them to the
capacity w/ and w/o the interference-power constraint, as well as
the achievable rate by a simple design for the secondary transmit
spatial spectrum, termed {\it white spectrum}, for which equal power
is allocated to all transmit antennas, i.e.,
$\mv{S}=\frac{P}{M_{t,s}}\mv{I}$. It is assumed that there is a
single-antenna primary receiver, $\gamma=0.1$, and
$M_{t,s}=M_{r,s}=2$. It is observed that the achievable rates by the
D-SVD and P-SVD are close to the capacity w/ interference-power
constraint at the low- and high-SNR regime, respectively. At the
high-SNR regime, both the capacity by the optimal $\mv{S}$ obtained
by numerical algorithm and the achievable rate by the P-SVD increase
linearly with $\log_2 SNR$ by a factor of $\min(M_{r,s},
M_{t,s}-1)=\min(2,2-1)=1$ while the achievable rate by the D-SVD
eventually gets saturated, both in accordance with Theorem
\ref{theorem:high P}. Substantial rate improvements by the optimal
as well as the D-SVD/P-SVD solutions over the simple white spectrum
are also observed from low to high SNR values.

\subsection{Hybrid D-SVD/P-SVD}
Fig. \ref{fig:SVD comp MU} compares the achievable rates of the
hybrid D-SVD/P-SVD with different values of $b$. The capacity w/ and
w/o the interference-power constraint are also shown for comparison.
It is assumed that the number of primary receivers $K=2$ each with a
single receive antenna, $\Gamma_k=0.1, k=1,2$, and
$M_{t,s}=M_{r,s}=4$. Because there are two single-antenna primary
receivers, the hybrid algorithm can choose to project the secondary
user's channel $\mv{H}$ into either the whole space of
$\hat{\mv{G}}^{\dag}$ by taking $b=2$ (or the P-SVD) or some
selected subspace of $\hat{\mv{G}}^{\dag}$ with $b=1$, or simply
choose not to project at all with $b=0$ (or the D-SVD). It is
observed that as SNR increases, the optimal SVD-based algorithm
changes from the D-SVD to the hybrid D-SVD/P-SVD with $b=1$, and
finally to the P-SVD. This result is consistent with the conjecture
we have made previously in Section \ref{subsubsection:Hybrid SVD}.
The rational here is that as the transmit power increases, the
secondary user's transmit spatial spectrum also needs to allocate a
larger portion of total transmit power into the projected channel in
order not to violate the prescribed interference-power constraint at
the primary receiver.

Fig. \ref{fig:SVD comp K} compares the achievable rates of the
hybrid D-SVD/P-SVD with the D-SVD and P-SVD as the number of
single-antenna primary receivers $K$ increases from 2 to 10. The
interference-power constraint $\Gamma_k=0.01$ is assumed to be equal
at all primary receivers, and $P_t$ is fixed as 10, and
$M_{t,s}=M_{r,s}=4$. Notice that for each given $K$, $b$ for the
hybrid algorithm can take values from 0 to $\min(M_{t,s},K)$. If
$b=0$, the hybrid algorithm becomes the D-SVD; while if
$b=\min(M_{t,s},K)$, it becomes the P-SVD. In this figure, the
achievable rate shown for the hybrid algorithm for each $K$ is
obtained by taking the maximum rate over all possible values of $b$.
Notice that if $K\geq M_{t,s}$, the achievable rate by the P-SVD
becomes zero because the ZF-based channel projection results in a
null projected channel $\mv{H}_{\bot}$ in (\ref{eq:channel
projection multiple primary receivers}). It is observed that as $K$
increases, the achievable rate by any of these three proposed
algorithms decreases due to more added interference-power
constraints. However, the hybrid algorithm by dynamically selecting
$b$ achieves substantial rate improvements compared with the D-SVD
and P-SVD because it effectively exploits the interference-diversity
gain in the CR network. As $K$ becomes much larger than $M_{t,s}$,
it is observed that the achievable rates by the hybrid algorithm and
the D-SVD both converge to be identical. However, such convergence
is observed to be very slow with respect to $K$.

\subsection{Multi-Channel Transmission}
A multi-tone or OFDM-based broadband transmission system is
considered for a secondary user with $M_{t,s}=M_{r,s}=2$ and a
single-antenna primary receiver. The OFDM channels for both primary
and secondary transmissions are assumed to have $N=64$ tones, and
four equal-energy, independent, and consecutive multi-path delays.
The interference-power constraint at the primary receiver is set to
be $\gamma=0.1$ at all tones. In this case, the secondary user is
able to allocate variable transmit rate, power, and spatial spectrum
at different OFDM tones based on the channels $\mv{H}_j$ and
$\mv{g}_j$, $j=1,\ldots,N$. Fig. \ref{fig:Multitone} shows the
secondary user's capacity w/ and w/o the interference-power
constraint, and the achievable rate of the per-tone-based suboptimal
algorithm that selects the best SVD-based algorithm (the D-SVD or
P-SVD) at each tone. It is observed that by exploiting the
frequency-selective fading, the per-tone-based suboptimal algorithm
achieves smaller gaps from the actual capacity w/ interference-power
constraints compared with the case of single-channel transmission
previously shown in Fig.\ref{fig:SVD comp SU}.

\section{Concluding Remarks and Future Directions} \label{sec:conclusion}
Transmit optimization for a single secondary MIMO/MISO link in a CR
network under constraint of opportunistic spectrum sharing is
considered. The capacity of the secondary link is studied under both
the secondary transmit-power constraint and a set of
interference-power constraints at multiple primary receivers.
Multi-antennas are exploited at the secondary transmitter to
optimally tradeoff between throughput maximization and interference
avoidance. In the case of MISO secondary user's channel, beamforming
is shown to be the optimal strategy for the secondary transmitter.
In the case of MIMO secondary user's channel, the capacity-achieving
transmit spatial spectrum under the interference-power constraint in
general does not diagonalize the MIMO channel and hence requires
sophisticated encoding and decoding methods. Two suboptimal
algorithms based on the SVD of the secondary user's MIMO channel,
namely the D-SVD and the P-SVD, are proposed to tradeoff between
capacity and complexity. In the case of multiple primary
receivers/antennas, a hybrid D-SVD/P-SVD algorithm is also proposed
to exploit the inherent interference diversity gain in a CR network.
The developed algorithms are also extended to the multi-channel
transmission whereby the secondary user is able to employ transmit
adaptations in all space, time and frequency domains so as to
achieve optimum opportunistic spectrum sharing.

An interesting extension of this paper is to consider more
generalized interference constraints at each primary receiver,
especially when it is equipped with multiple receive antennas.
Instead of considering either the per-antenna-based power constraint
or the total-power constraint over all receive antennas, the
interference constraint at each primary receiver can also be more
specifically related to the primary user's channel capacity. More
interestingly, both primary and secondary transmit adaptations can
be jointly optimized under mutual interferences given their unequal
priorities for spectrum utilization, through either a centralized
control or distributed self-adaptations based on the principle of
game theory. Results and algorithms developed in this paper can also
be extended to the case where only statistical knowledge on the
channel, instead of instantaneous channel knowledge as assumed in
this paper, is available at the secondary transmitter for optimum
design of transmit spatial spectrum under some long-term average
interference-power constraints. Furthermore, although this paper
addresses a single pair of secondary transmitter and receiver, it
can also be extended to incorporate multiple secondary users that
jointly share transmit spectrum with the primary users by
considering different models of the secondary network, e.g., the
uplink MIMO multiple-access channel (MAC), the downlink MIMO
broadcast channel (BC), or the distributed MIMO interference channel
(IC).

\appendices
\section{Proof of Theorem \ref{theorem:capacity loss} } \label{appendix:capacity loss}
Consider the $k$-th primary transmission represented by
\begin{equation}
\mv{y}_k(n) = \mv{H}_k\mv{x}_k(n) + \mv{z}_k(n) + \mv{q}_k(n),
\end{equation}
where $\mv{H}_k\in\mathbb{C}^{M_{k} \times N_{k}}$ denotes the
$k$-th primary user's channel, $M_{k}$ and $N_{k}$ are the number of
antennas at its receiver and transmitter, respectively;
$\mv{y}_k(n)$ and $\mv{x}_k(n)$ are the received and transmitted
signal vector, respectively; $\mv{q}_k(n)$ is the additive noise
vector at the receiver that assumed to be distributed as
$\mv{z}_k(n)\sim\mathcal{CN}(0,\phi_k\mv{I})$; $\mv{q}_k(n)$ is the
interference from the secondary transmitter expressed as
\begin{equation}
\mv{q}_k(n) = \mv{G}_k\mv{x}(n).
\end{equation}
It is easy to verify that the covariance matrix for $\mv{q}_k(n)$,
denoted by $\mv{Q}_k=\mathbb{E}[\mv{q}_k(n)\mv{q}_k^{\dag}(n)]$, is
equal to $\mv{G}_k\mv{S}\mv{G}_k^{\dag}$, where $\mv{S}$ is the
transmit covariance matrix for the secondary user. Let the transmit
covariance matrix of the $k$-th primary user be denoted by
$\mv{S}_k$, $\mv{S}_k=\mathbb{E}[\mv{x}_k(n)\mv{x}_k^{\dag}(n)]$.
The capacity of the $k$-th primary transmission (in bits/complex
dimension) can be then expressed as
\begin{equation}
C_1=\log_2\left|\mv{I}+\left(\phi_k\mv{I}+\mv{Q}_k\right)^{-1/2}\mv{H}_k\mv{S}_k\mv{H}_k^{\dag}
\left(\phi_k\mv{I}+\mv{Q}_k\right)^{-1/2}\right|.
\end{equation}
On the other hand, if the secondary transmitter if off, i.e., the
interference from the secondary transmitter $\mv{q}_k(n)$ is
nonexistent, the capacity of the $k$-th primary transmission is
expressed as
\begin{equation}
C_2=\log_2\left|\mv{I}+\frac{1}{\phi_k}\mv{H}_k\mv{S}_k\mv{H}_k^{\dag}\right|.
\end{equation}
The capacity loss of the $k$-th primary transmission due to the
secondary transmission is then equal to $C_2-C_1$. In order to find
an upper-bound for such capacity loss, the following
equalities/ineuqalities are provided for $C_1$:
\begin{eqnarray}
C_1&=&\log_2\left|\mv{I}+\mv{S}^{1/2}_k\mv{H}_k^{\dag}
\left(\phi_k\mv{I}+\mv{Q}_k\right)^{-1}
\mv{H}_k\mv{S}_k^{1/2}\right|, \label{eq:a}\\
&\geq&
\log_2\left|\mv{I}+\frac{1}{\phi_k+\Gamma_k}\mv{S}^{1/2}_k\mv{H}_k^{\dag}
\mv{H}_k\mv{S}_k^{1/2}\right|, \label{eq:b} \\
&=&
\log_2\left|\mv{I}+\frac{1}{\phi_k+\Gamma_k}\mv{H}_k\mv{S}_k\mv{H}_k^{\dag}\right|,
\label{eq:c}\\
&\geq&
\log_2\left|\frac{\phi_k}{\phi_k+\Gamma_k}\left(\mv{I}+\frac{1}{\phi_k}\mv{H}_k\mv{S}_k\mv{H}_k^{\dag}\right)\right|,\\
&=&-\min(M_k,N_k)\log_2\left(1+\frac{\Gamma_k}{\phi_k}\right)+C_2,
\label{eq:d}
\end{eqnarray}
where (\ref{eq:a}) and (\ref{eq:c}) are due to
$\log|\mv{I}+\mv{AB}|=\log|\mv{I}+\mv{BA}|$, (\ref{eq:b}) is from
the facts that $\mv{Q}_k\preceq\Gamma_k\mv{I}$ and
$\log|\mv{I}+\mv{A}\mv{X}_1\mv{A}^{\dag}|\geq
\log|\mv{I}+\mv{A}\mv{X}_2\mv{A}^{\dag}|$, if
$\mv{X}_1\succeq\mv{X}_2\succeq 0$. Using (\ref{eq:d}), the proof is
completed.

\section{Proof of Lemma \ref{lemma:MISO 1}}\label{appendix:proof of
lemma 1}

Because Problem {\bf P2} is convex, the optimal $\mv{S}$ must
satisfy the Karush-Kuhn-Tacker (KKT) conditions \cite{Boydbook}
written at below:
\begin{eqnarray}
\mv{h}^{\dag}\left(1+\mv{h}\mv{S}\mv{h}^{\dag}\right)^{-1}\mv{h}+\mv{\Phi}&=&
\mu\mv{g}^{\dag}\mv{g}+\nu\mv{I}, \label{eq:KKT 1}
\\ \nu\left(\mathtt{Tr}(\mv{S})-P_t\right)&=&0, \label{eq:KKT 2}
\\ \mu\left(\mv{g}\mv{S}\mv{g}^{\dag}- \gamma\right)&=&0,
\label{eq:KKT 3}
\\ \mathtt{Tr}(\mv{\Phi}\mv{S})&=&0, \label{eq:KKT 4}
\end{eqnarray}
where $\nu$, $\mu$, and $\mv{\Phi}$ are the Lagrange multipliers
associated with the constraint (\ref{eq:Tx power}), (\ref{eq:Ic
power}) and (\ref{eq:S}), respectively, $\nu, \mu \geq 0,
\mv{\Phi}\succeq 0$. First, consider the case of $\nu=0$, from
(\ref{eq:KKT 2}), it follows that $P<P_t$, i.e., $P$ is not limited
by the secondary transmit-power constraint, but instead, by the
interference-power constraint. From (\ref{eq:KKT 1}) and because
$\mv{\Phi}\succeq 0$, it follows that $\mv{g}$ must be in parallel
to $\mv{h}$, i.e., $\mv{g}=c\mv{h}$, where $c$ is a constant. In
this case, it can be easily shown that the optimal $\mv{S}$ is
$\mv{S}=\frac{\gamma}{||\mv{g}\mv{h}^{\dag}||^2}\mv{h}^{\dag}\mv{h}$,
i.e., $\mathtt{Rank}(\mv{S})=1$. Next, consider the case of $\nu>0$,
i.e., $P=P_t$. In this case, the RHS of (\ref{eq:KKT 1}) has a full
rank of $M_{t,s}$ regardless of $\mu$. It then follows that at the
LHS of (\ref{eq:KKT 1}), since the first term has a unit rank,
$\mathtt{Rank}(\mv{\Phi})\geq M_{t,s}-1$. Since $\mv{S}\succeq 0$
and $\mv{\Phi}\succeq 0$, from (\ref{eq:KKT 4}) it follows that
$\mathtt{Rank}(\mv{S})+\mathtt{Rank}(\mv{\Phi})\leq M_{t,s}$. Hence,
$\mathtt{Rank}(\mv{S})\leq 1$, and the proof is completed.

\section{Proof of Lemma \ref{lemma:MISO 2}}\label{appendix:proof of
lemma 2}

Let the optimal beamforming vector $\mv{v}'=\alpha_{v'}\hat{\mv{g}}
+ \beta_{v'}\hat{\mv{b}} $, where
$\hat{\mv{b}}^{\dag}\hat{\mv{g}}=0$. It can be then verified that
replacing $\hat{\mv{b}}$ with $\hat{\mv{h}}_{\bot}$ does not
increase the interference power $||\mv{g}\mv{v}||^2$ in (\ref{eq:Ic
power MISO}) since $(\hat{\mv{h}}_{\bot})^{\dag}\hat{\mv{g}}=0$, but
always helps to improve the SNR at the secondary receiver
$||\mv{h}\mv{v}||^2$ in (\ref{eq:capacity MISO}) since
$||\mv{h}\hat{\mv{h}}_{\bot}||\geq ||\mv{h}\hat{\mv{b}}||$. The
proof is then completed by the fact that the secondary user's
capacity always increases with the receiver SNR.

\section{Proof for Optimality of P-SVD when $\gamma=0$}\label{appendix:optimality of P-SVD}
We show that if in the interference-power constraint (\ref{eq:Ic
power}) of Problem {\bf P2}, $\gamma=0$, the P-SVD is indeed
optimal. First, using the constraint
$\mv{g}\mv{S}\mv{g}^{\dag}=\gamma=0$, (\ref{eq:capacity}) can be
simplified as
\begin{eqnarray}
&& \log_2\left|\mv{I}+\mv{H}\mv{S}\mv{H}^{\dag}\right| \\
&=& \log_2\left|\mv{I}+(\mv{H}_{\bot}+
\mv{H}\hat{\mv{g}}\hat{\mv{g}}^{\dag})\mv{S}(\mv{H}_{\bot}+
\mv{H}\hat{\mv{g}}\hat{\mv{g}}^{\dag})^{\dag}\right|, \label{eq:i} \\
&=&\log_2\left|\mv{I}+\mv{H}_{\bot}\mv{S}(\mv{H}_{\bot})^{\dag}\right|,
\label{eq:ii}
\end{eqnarray}
where (\ref{eq:i}) is from (\ref{eq:channel projection}). Using
(\ref{eq:ii}), Problem {\bf P2} with $\gamma=0$ can be rewritten as
\begin{eqnarray}
\mathtt{Maximize}&&
\log_2\left|\mv{I}+\mv{H}_{\bot}\mv{S}(\mv{H}_{\bot})^{\dag}\right|
\\ \mathtt {Subject \ to} && \mathtt{Tr}(\mv{S})\leq P_t, \\  && \mv{g}\mv{S}\mv{g}^{\dag}=0, \label{eq:Ic power appendix} \\ && \mv{S}\succeq 0.
\end{eqnarray}
Notice that without the constraint (\ref{eq:Ic power appendix}), the
optimal $\mv{S}^*$ for the above problem should have the same form
of $\mv{U}_{\bot}\mv{\Sigma}(\mv{U}_{\bot})^{\dag}$ as by the P-SVD.
At last, it remains to check whether
$\mv{g}\mv{S}^*\mv{g}^{\dag}=0$. This is so because
$\hat{\mv{g}}^{\dag}\mv{U}_{\bot}=0$.

\section{Proof of Lemma \ref{lemma:multiplexing gain}}\label{appendix:proof of lemma 3}

The following equalities/inequalities hold:
\begin{eqnarray}
&& \log_2\left|\mv{I}+\mv{H}\mv{S}\mv{H}^{\dag}\right| \\ &=&
\log_2\left|\mv{I}+(\mv{H}_{\bot}+
\mv{H}\hat{\mv{g}}\hat{\mv{g}}^{\dag})\mv{S}(\mv{H}_{\bot}+
\mv{H}\hat{\mv{g}}\hat{\mv{g}}^{\dag})^{\dag}\right|, \label{eq:1}\\
&\leq&\log_2\left|\mv{I}+2\mv{H}_{\bot}\mv{S}(\mv{H}_{\bot})^{\dag}+2
\mv{H}\hat{\mv{g}}\hat{\mv{g}}^{\dag}\mv{S}(
\mv{H}\hat{\mv{g}}\hat{\mv{g}})^{\dag}\right| \label{eq:2} \\ &\leq&
\log_2\left|\mv{I}+2\mv{H}_{\bot}\mv{S}(\mv{H}_{\bot})^{\dag}\right|+\log_2\left|\mv{I}+2
\mv{H}\hat{\mv{g}}\hat{\mv{g}}^{\dag}\mv{S}(
\mv{H}\hat{\mv{g}}\hat{\mv{g}})^{\dag}\right| \label{eq:3},
\end{eqnarray}
where (\ref{eq:1}) is from (\ref{eq:channel projection}),
(\ref{eq:2}) is by
$(\mv{A}+\mv{B})\mv{S}(\mv{A}+\mv{B})^{\dag}\preceq
2(\mv{A}\mv{S}\mv{A}^{\dag}+\mv{B}\mv{S}\mv{B}^{\dag})$ for
$\mv{S}\succeq 0$, (\ref{eq:3}) is  by
$\log|\mv{I}+\mv{S}_1+\mv{S}_2|\leq \log|\mv{I}+\mv{S}_1| +
\log|\mv{I}+\mv{S}_2|$ for $\mv{S}_1, \mv{S}_2\succeq 0$. Let
$\mathcal{D}_S$ denote the set for $\mv{S}$ specified by (\ref{eq:Tx
power}), (\ref{eq:Ic power}) and (\ref{eq:S}), it then follows that
\begin{eqnarray}
R_{\rm opt} &=& \max_{\mv{S}\in\mathcal{D}_S}
\log_2\left|\mv{I}+\mv{H}\mv{S}\mv{H}^{\dag}\right| \\ &\leq&
\max_{\mv{S}\in\mathcal{D}_S}\log_2\left|\mv{I}+2\mv{H}_{\bot}\mv{S}(\mv{H}_{\bot})^{\dag}\right|+\mathcal{O}(1), \label{eq:4} \\
&\leq&
\max_{\mv{S}\in\mathcal{D}_S}\log_2\left|2\left(\mv{I}+\mv{H}_{\bot}\mv{S}(\mv{H}_{\bot})^{\dag}\right)\right|+\mathcal{O}(1),
\\ &\leq& R_{\rm P-SVD}+M_{r,s}+\mathcal{O}(1) \label{eq:5},
\end{eqnarray}
where (\ref{eq:4}) follows from (\ref{eq:3}) by noticing that due to
(\ref{eq:Ic power}) the second term on the RHS of (\ref{eq:3}) is
upper-bounded  by some finite constant. From (\ref{eq:5}) and
because $R_{\rm opt}\geq R_{\rm P-SVD}$, it follows that
$\lim_{P_t\rightarrow \infty}\frac{R_{\rm
opt}}{\log_2P_t}=\lim_{P_t\rightarrow \infty}\frac{R_{\rm
P-SVD}}{\log_2P_t}$. Since it has already been verified in the main
text that $\lim_{P_t\rightarrow \infty}\frac{R_{\rm
P-SVD}}{\log_2P_t}=\min(M_{r,s},M_{t,s}-1)$, the proof is completed.

\newpage

\begin{figure}
\psfrag{a}{Secondary Transmitter}\psfrag{b}{Secondary Receiver}
\psfrag{c}{Primary Receiver 1}\psfrag{d}{Primary Receiver K}
\psfrag{e}{$\mathtt{Tr}\left(\mv{S}\right)\leq
P_t$}\psfrag{f}{$\mathtt{Tr}\left(\mv{G}_1\mv{S}\mv{G}_1^{\dag}\right)\leq
\Gamma_1$}\psfrag{g}{$\mathtt{Tr}\left(\mv{G}_K\mv{S}\mv{G}_K^{\dag}\right)\leq
\Gamma_K$}\psfrag{h}{$\mv{H}$}\psfrag{i}{$\mv{G}_1$}\psfrag{j}{$\mv{G}_K$}\psfrag{k}{Primary
Receiver
2}\psfrag{l}{$\mv{G}_2$}\psfrag{m}{$\mathtt{Tr}\left(\mv{G}_2\mv{S}\mv{G}_2^{\dag}\right)\leq
\Gamma_2$}
\begin{center}
\scalebox{0.8}{\includegraphics*[76pt,490pt][393pt,722pt]{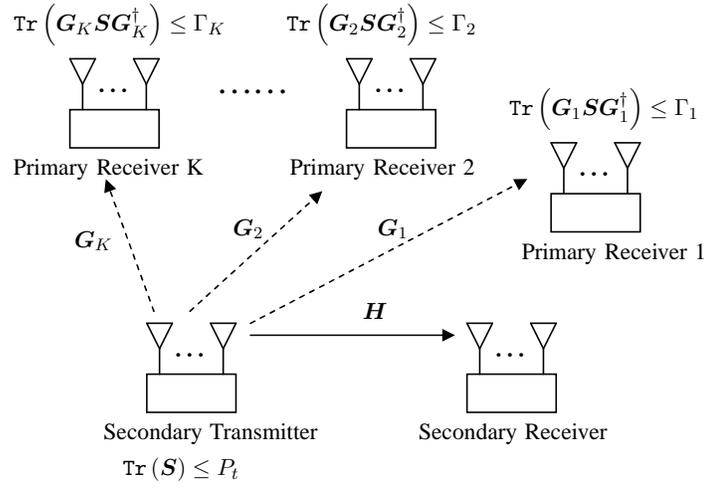}}
\end{center}
\caption{A cognitive radio (CR) network where the secondary user
shares the same transmit spectrum with $K$ primary
users.}\label{fig:CR system}
\end{figure}

\begin{figure}
\centering{
 \epsfxsize=4.0in
    \leavevmode{\epsfbox{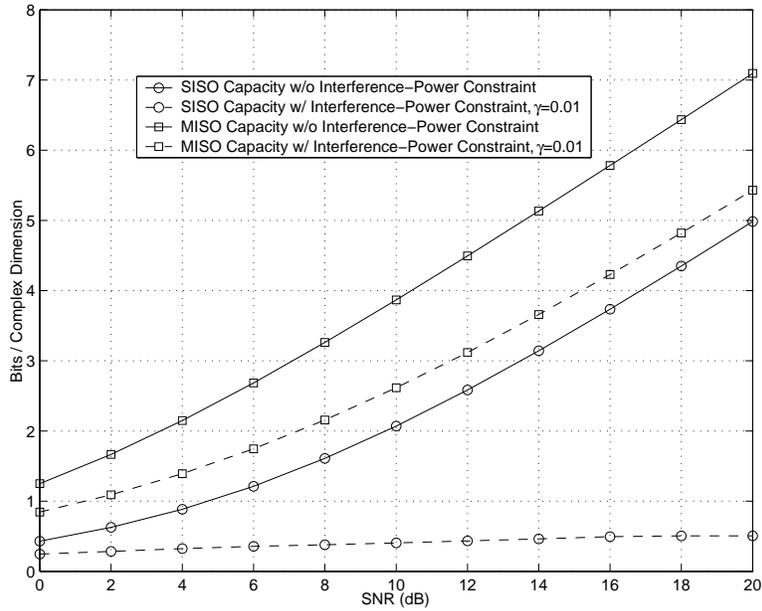}} }
\caption{Capacity comparison for the secondary transmission w/ and
w/o the interference-power constraint at a single-antenna primary
receiver.}\label{fig:capacity comp SU}
\end{figure}

\begin{figure}
\centering{
 \epsfxsize=4.0in
    \leavevmode{\epsfbox{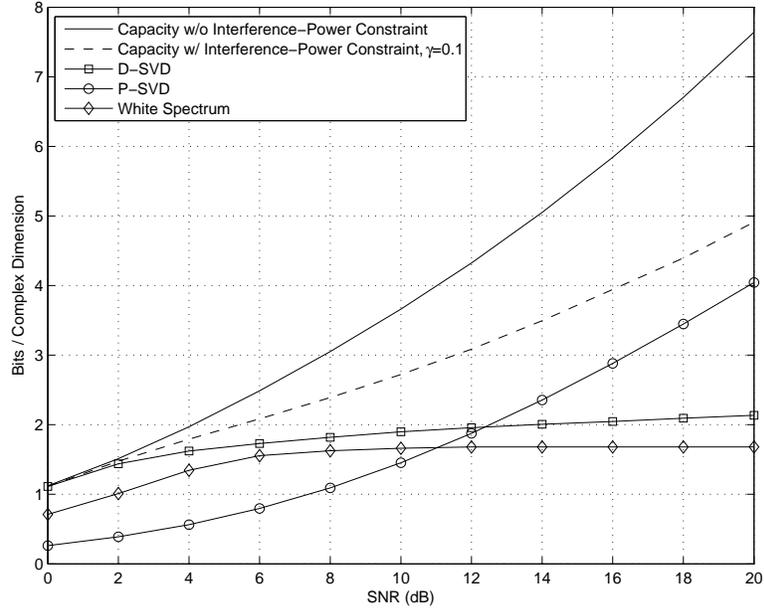}} }
\caption{Comparison of the achievable rates for the secondary
transmission by the D-SVD and P-SVD in the case of a single-antenna
primary receiver, and $M_{t,s}=M_{r,s}=2$.}\label{fig:SVD comp SU}
\end{figure}

\begin{figure}
\centering{
 \epsfxsize=4.0in
    \leavevmode{\epsfbox{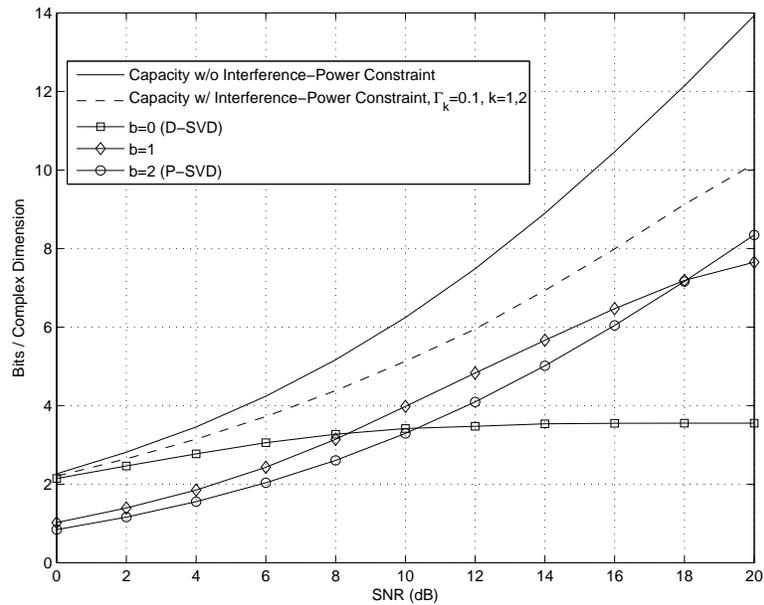}} }
\caption{Comparison of the achievable rates for the secondary
transmission by the hybrid D-SVD/P-SVD with different values of $b$,
in the case of two single-antenna primary receivers, and
$M_{t,s}=M_{r,s}=4$.}\label{fig:SVD comp MU}
\end{figure}

\begin{figure}
\centering{
 \epsfxsize=4.0in
    \leavevmode{\epsfbox{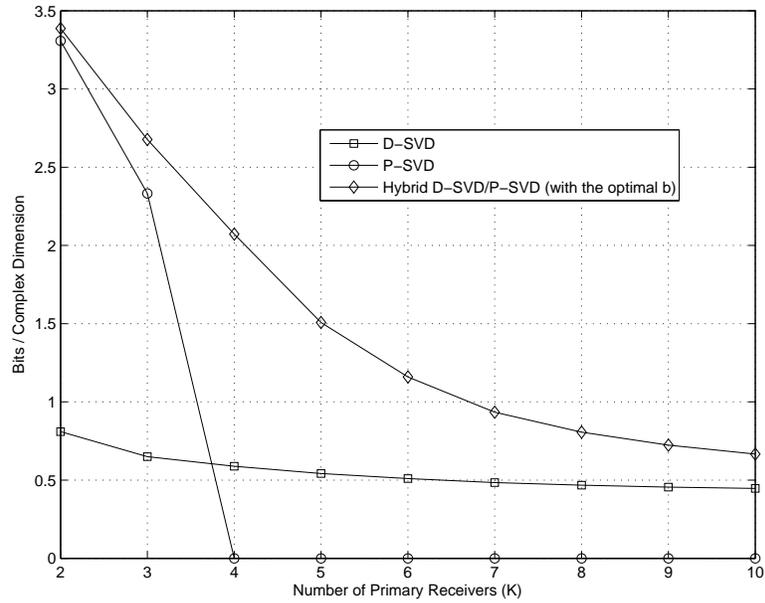}} }
\caption{Comparison of the achievable rates for the secondary
transmission by the hybrid D-SVD/P-SVD, the D-SVD and the P-SVD for
different number of single-antenna primary receivers $K$, in the
case of $M_{t,s}=M_{r,s}=4$ and a constant secondary power
constraint $P_t=10$.}\label{fig:SVD comp K}
\end{figure}

\begin{figure}
\centering{
 \epsfxsize=4.0in
    \leavevmode{\epsfbox{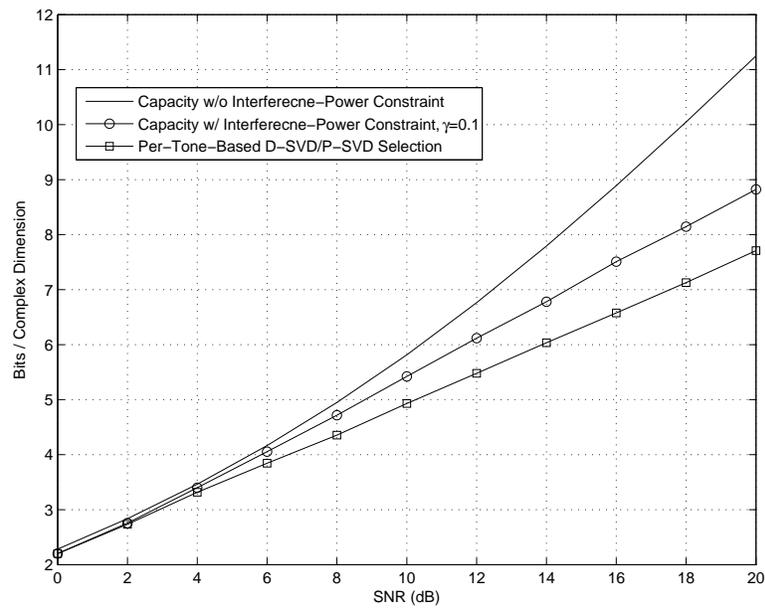}} }
\caption{Comparison of the achievable rates for the secondary
transmission in a multi-tone-based broadband system with $N=64$, a
single-antenna primary receiver, and
$M_{t,s}=M_{r,s}=2$.}\label{fig:Multitone}
\end{figure}

\end{document}